\documentclass[lettersize,journal]{IEEEtran}
\usepackage{amsmath,amsfonts}
\usepackage{algorithmic}
\usepackage{algorithm}
\usepackage{array}
\usepackage[caption=false,font=footnotesize,labelfont=rm,textfont=rm]{subfig}
\usepackage{textcomp}
\usepackage{stfloats}
\usepackage{url}
\usepackage{verbatim}
\usepackage{graphicx}
\usepackage{cite}
\usepackage{amssymb}
\usepackage{caption}
\usepackage{booktabs}
\usepackage{amsthm}
\usepackage{amsmath}
\newtheorem{lemma}{Lemma}
\newtheorem{theorem}{Theorem}
\usepackage{subcaption}
\begin{document}


\title{Hierarchical Online-Scheduling for Energy-Efficient Split Inference with Progressive Transmission}


\author{Zengzipeng Tang,~\IEEEmembership{Student Member,~IEEE,} Yuxuan Sun,~\IEEEmembership{Member,~IEEE,} Wei Chen,~\IEEEmembership{Senior Member,~IEEE,}
Jianwen Ding, Bo Ai,~\IEEEmembership{Fellow,~IEEE} and Yulin Shao,~\IEEEmembership{Member,~IEEE}
\thanks{Zengzipeng Tang, Yuxuan Sun (Corresponding Author), Wei Chen, Jianwen Ding and Bo Ai are with the School of Electronic and Information Engineering, Beijing Jiaotong University, Beijing 100044, China. (e-mail: \{zzptang, yxsun, weich, jwding, boai\}@bjtu.edu.cn)}
\thanks{Yulin Shao is with the Department of Electrical and Electronic Engineering, The University of Hong Kong (e-mail: ylshao@hku.hk).}
}

\markboth{Journal of \LaTeX\ Class Files,~Vol.~14, No.~8, August~2021}%
{Shell \MakeLowercase{\textit{et al.}}: A Sample Article Using IEEEtran.cls for IEEE Journals}


\maketitle

\begin{abstract}
Device-edge collaborative inference with Deep Neural Networks (DNNs) faces fundamental trade-offs among accuracy, latency and energy consumption. Current scheduling exhibits two drawbacks: a granularity mismatch between coarse, task-level decisions and fine-grained, packet-level channel dynamics, and insufficient awareness of per-task complexity. Consequently, scheduling solely at the task level leads to inefficient resource utilization.
This paper proposes a novel ENergy-ACcuracy Hierarchical optimization framework for split Inference, named ENACHI, that jointly optimizes task- and packet-level scheduling to maximize accuracy under energy and delay constraints. A two-tier Lyapunov-based framework is developed for ENACHI, with a progressive transmission technique further integrated to enhance adaptivity. At the task level, an outer drift-plus-penalty loop makes online decisions for DNN partitioning and bandwidth allocation, and establishes a reference power budget to manage the long-term energy-accuracy trade-off. At the packet level, an uncertainty-aware progressive transmission mechanism is employed to adaptively manage per-sample task complexity. This is integrated with a nested inner control loop implementing a novel reference-tracking policy, which dynamically adjusts per-slot transmit power to adapt to fluctuating channel conditions. Experiments on ImageNet dataset demonstrate that ENACHI outperforms state-of-the-art benchmarks under varying deadlines and bandwidths, achieving a 43.12\% gain in inference accuracy with a 62.13\% reduction in energy consumption under stringent deadlines, and exhibits high scalability by maintaining stable energy consumption in congested multi-user scenarios. 
\end{abstract}


\section{Introduction}
\IEEEPARstart{T}{he} deployment of Deep Neural Networks (DNNs) on smart mobile devices is driving a fundamental shift toward edge Artificial Intelligence (AI), essential for supporting real-time inference in applications like augmented reality (AR), virtual reality (VR), and advanced perception networks for the Industrial Internet of Things (IIoT). These applications involve continuous AI task generation and require high-throughput visual perception, creating a significant demand for on-device, context-aware intelligence \cite{edge3.5}. Concurrently, the rise of paradigms like Integrated Sensing and Edge AI (ISEA) \cite{isea} further intensifies this trend, demanding an unprecedented level of tight coupling between environmental sensing, real-time computation, and ultra-reliable communication to deliver context-aware intelligence \cite{edge4}, 
which is increasingly viewed as a key enabler for 6G intelligent perception.
This trend, however, faces a core conflict with the physical limitations of edge devices. The pursuit of state-of-the-art performance has led to an exponential increase in model size and computational complexity \cite{2019}. This creates a major energy challenge for mobile devices \cite{conflicts}, which often have severe constraints on battery capacity and processing power, hindering their ability to independently run large-scale AI models.



To tackle this, the device-edge collaborative paradigm of \textit{split inference}, enabled by Edge Intelligence (EI) \cite{edgeai}, has emerged as a promising solution. By partitioning a DNN between the device and the edge server, this architecture offloads computation-intensive layers to the edge, thereby alleviating the processing load of the device. However, it does not fundamentally resolve the energy and latency challenges, as both local computation and wireless transmission remain power-consuming. Moreover, task success is often constrained by strict real-time requirements. In latency-sensitive applications such as AR and industrial automation \cite{IMT}, inference results lose their utility once deadlines are missed, turning latency from a soft performance metric into a hard Quality-of-Service (QoS) constraint.

Optimizing inference accuracy and energy performance for deadline-critical AI tasks is highly challenging. Key operational decisions, such as where to partition the DNN, how to allocate wireless resources, and how much intermediate feature data to transmit, create a \emph{complex interplay} between inference accuracy, end-to-end latency, and device energy consumption. 
An attempt to improve one metric, such as increasing accuracy by sending more data, can negatively impact others by increasing energy use and potentially violating the task deadline. 
Compounding this difficulty is the \emph{non-analytical} nature of DNNs, where the specific mapping of inference accuracy and energy consumption lacks a closed-form expression.
Such inherent conflicts of tight coupling and model intractability render the optimization of the energy-accuracy trade-off exceptionally difficult. This challenge is further magnified in a multi-user EI system, where multiple devices must simultaneously compete for limited wireless bandwidth and shared edge-server computational resources. 
Recent studies have explored joint optimization frameworks to balance latency, energy consumption, and QoS in edge inference. These frameworks build upon pioneering research that established task offloading through online resource optimization \cite{p1,p2,p21} and strategies for model pruning or partitioning to reduce computation and communication costs \cite{p4, p3}. Motivated by the pressing demand for energy efficiency in green edge AI \cite{green}, representative efforts include dynamic partitioning schemes with early exits via reinforcement learning \cite{c2} and Lyapunov-based methods for long-term energy and latency management through joint partitioning and resource allocation \cite{c3}. Further advancements investigate semantic communication-based offloading \cite{c4} and layer-level DNN partitioning with distributed game-theoretic optimization \cite{c5}.
In multi-user scenarios, batching strategies are widely employed to enhance server throughput \cite{batching}, and similar batch-processing techniques are utilized to reduce collective energy overheads \cite{shi}. To address the challenge of performance estimation, \cite{han} focuses on performance modeling for energy-aware scheduling. In video analytics and AR, researchers model the relationship between transmitted data and inference accuracy using curve fitting or configuration adaptation to address the non-analytical nature of DNNs\cite{29, 23, c1}.

Recent advancements further broaden the optimization scope to environmental dynamics and new system paradigms. Considering the dynamic nature of edge environments, robust optimization is introduced to address uncertainties in computing capacity \cite{nan1} and inference delay \cite{nan2}. In the emerging paradigm of ISEA, research targets ultra-low-latency edge inference for distributed sensing \cite{point1ms}. Complementary studies also investigate event-triggered offloading strategies to effectively mitigate communication overheads \cite{zhoua}. Furthermore, hierarchical scheduling is employed to separates long-term energy provisioning from short-term task offloading decisions in multi-tier edge systems \cite{new1}. Similarly, a two-timescale approach in satellite edge computing \cite{new2} decouples slow service deployment optimization from fast task scheduling.

Despite notable progress in comprehensive resource and quality management, contemporary split-inference frameworks still face two critical challenges:

(1) \textit{Mismatch in scheduling granularity}.
The high dimensionality of intermediate representations in modern DNNs, such as the 256 × 56 × 56 feature maps in ResNet, totaling over 8 million real-valued parameters \cite{parameter}, creates a severe communication bottleneck over bandwidth-limited wireless links. To alleviate this issue, a viable approach is to employ a \textit{progressive feature transmission} strategy\cite{progressive}, wherein intermediate feature representations are transmitted in multiple adaptive stages. At each stage, the edge server selectively receives the most informative feature dimensions and decides whether to continue transmission based on feedback about accuracy improvement or channel quality. This fine-grained, feedback-driven mechanism significantly improves transmission efficiency and enables slot-wise adaptation to dynamic wireless conditions, which naturally aligns with the fine-grained, slot-by-slot nature of wireless resource allocation. While the mechanism improves communication efficiency, integrating such packet-level transmission mechanism with macro task-level resource allocation remains difficult, revealing a fundamental mismatch in scheduling granularity.

(2) \textit{Lack of task-aware adaptation}. Most existing split-inference frameworks adopt uniform offloading strategies that treat all inference tasks equally, neglecting the inherent heterogeneity among them. In practice, tasks often differ in importance, input complexity, and feature representation sparsity, resulting in distinct transmission priorities and accuracy-energy trade-offs. For instance, tasks with higher semantic importance or denser features may require deeper layer transmission or finer feature granularity, whereas simpler tasks can be satisfied with partial feature subsets \cite{guolei3,guolei}. Without explicitly modeling such task-specific variations, current frameworks fail to allocate communication and computation resources efficiently across heterogeneous workloads. This lack of task-aware adaptation ultimately restricts both inference accuracy and energy efficiency in dynamic edge environments.

In this paper, we consider a multi-user EI system, where users collaborate with the edge server for split inference tasks. To cope with the constraint of scarce wireless and local computation resources, we propose a novel ENergy-ACcuracy Hierarchical optimization for split-Inference framework, named ENACHI. We formulate an optimization problem to maximize the aggregate inference accuracy of all users, while ensuring energy consumption stability and satisfying hard deadline constraints. 
The framework uses a two-tier Lyapunov architecture, integrating progressive transmission to handle varying task complexity within the stochastic optimization. At the task level, an outer loop decides DNN partitioning and bandwidth allocation while setting a reference power budget. At the packet level, a nested inner loop employs an uncertainty-aware progressive transmission mechanism, integrated with a reference-tracking policy to adjust per-slot transmit power for coping with channel fluctuations.

The main contributions are summarized as follows:

\begin{itemize}

    \item We propose a novel hierarchical optimization framework ENACHI to address the joint optimization of long-term energy consumption and inference accuracy  in deadline-critical split inference, which coordinates task-level scheduling and packet-level transmission by employing a Lyapunov-based \textit{nested drift-plus-penalty} method and integrates adaptive progressive transmission with batched inference.

    \item To address the intracable nature of long-term stochastic problem and the non-analytical and non-convex inference accuracy objective, we develop a \textit{tractable surrogate model}, which enables efficient online task-level scheduling while providing \textit{theoretical performance guarantees} for long-term system stability and performance.

    \item We design a nested inner control loop to solve the scheduling granularity mismatch between task-level decisions and packet-level channel dynamics, which implements a \textit{task-aware}, \textit{reference-tracking policy} that tracks the power budget from task-level to ensure consistency, while adapting to channel conditions and task complexity.

    \item We conduct extensive simulations on the ImageNet dataset\cite{ima} to validate our framework. The results demonstrate that ENACHI outperforms benchmarks by achieving a favorable energy-accuracy trade-off and robust scalability, especially under stringent deadlines and in congested multi-user scenarios.
\end{itemize}


\section{System Model}\label{Section:system model}
We consider a device-edge collaborative edge AI system, as illustrated in Fig.~\ref{fig:system},  where $n$ computation-limited devices collected in a set of $\mathcal{N}=\{1,...,N\}$, cooperate with an edge server to execute DNN-based inference tasks. For flexibility, pre-trained DNN models of different types are cached in both the devices and the edge server. 

\begin{figure*}[htbp!]
    \centering
    \includegraphics[width=0.94\textwidth]{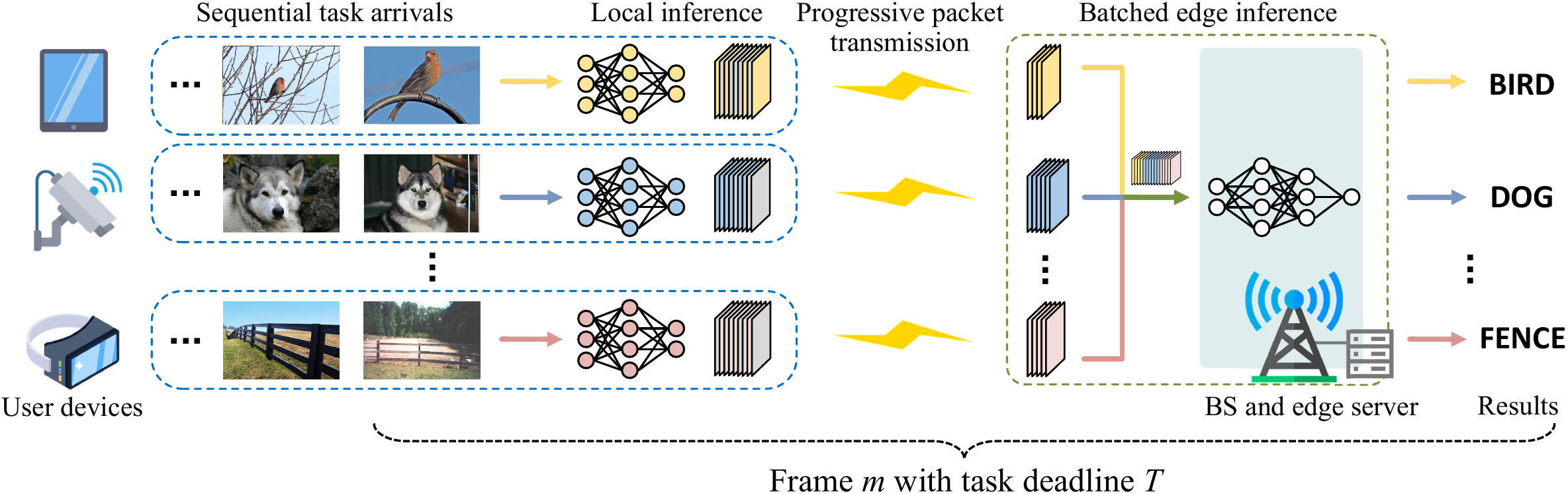} 
    \caption{Illustration of the multi-user split inference system. }
    \label{fig:system}
    \vspace{-0.2cm}
\end{figure*}

The system time is structured into discrete frames. Each task arrives at the beginning of frame $m \in \mathcal{M}$, where $\mathcal{M} = \{1, 2, \dots, M\}$, with a deadline $T$. Typical use cases include AR, where sequential image or video frames necessitate periodic processing to overlay digital information, or industrial IoT, where automated inspections must be performed on items on a production line at regular intervals. 
Consistent with these dedicated application scenarios, we assume all users execute tasks of a homogeneous type, i.e., leveraging an identical DNN architecture, throughout the entire operational period. Within each frame, a device may perform inference locally or engage in collaborative processing with the edge server. 
\subsection{Split Inference Model} \label{iia}
The DNN model of task in frame $m$ consists of $k_m$ sequential layers. For device $n$, when collaborative inference is performed, the model is partitioned at a layer $s_{n,m}$ selected from a predefined set of feasible partition points, $\mathcal{S}=\{0, 1, 2, \dots, k_m\}$, such that layers $0$ to $s_{n,m}$ are executed locally on the device, while layers $s_{n,m}+1$ through $k_m$ are offloaded and executed at the edge server. Specifically, $s_{n,m}=0$ indicates that the entire task is fully offloaded to the edge server for execution, whereas $s_{n,m}=k_m$ corresponds to fully local execution without offloading. 

The computational complexity of AI inference tasks is quantified in terms of the number of multiply-accumulate operations (MACs) \cite{c1}. For a given task in frame $m$ with partition point $s_{n,m}$, the total task workload is denoted by $R_{m}$ (in MACs), and the number of intermediate feature maps is represented as $b_\text{total}(s_{n,m})$.  Typically, the size of the final output is far smaller than the input or intermediate features and can thus be neglected in transmission considerations.\cite{neglect}.

We denote $f_{n,m}$ as the local computing ability of the $n$-th user device during frame $m$. The correspondingly local computing delay on frame $m$ can be written as:
\begin{equation}
    \setlength\abovedisplayskip{5pt}
    \setlength\belowdisplayskip{5pt}
    t_{n,m}^\text{local} = \frac {R_{s_{n,m}}^\text{local}}{f_{n,m}},
\end{equation}
where $R_{s_{n,m}}^\text{local}$ represents the local computing workload of device $n$ when choosing partition point $s_{n,m}$. The corresponding energy consumption is 
\begin{equation}
    \setlength\abovedisplayskip{5pt}
    \setlength\belowdisplayskip{5pt}
    E_{n,m}^\text{local}=\alpha_{n} f_{n,m}^3 t_{n,m}^\text{local} = \alpha_{n} f_{n,m}^2 R_{s_{n,m}}^\text{local},
\end{equation}
where $\alpha_{n}$ is a coefficient determined by the corresponding device chip architecture \cite{chip}.

After completing the local computation portion of the inference task, each element of feature maps is quantized with a sufficiently high resolution of $D$ bits, rendering quantization errors negligible. 
We divide task transmission time into time slots of duration $t_\text{slot}$ seconds (typically 1ms), which represent the typical scheduling interval in contemporary wireless communication systems \cite{3gpp}, and the wireless channel is assumed to remain constant within each time slot.
For each frame $m$, device $n$ is allocated a dedicated narrowband channel with an allocated bandwidth of $\omega_{n,m}$ that remains constant across all slots. Correspondingly, its achievable transmission rate in slot $k$ is given by the Shannon capacity formula:
\begin{equation}
    \setlength\abovedisplayskip{5pt}
    \setlength\belowdisplayskip{5pt}
    r_{n,m,k} = \omega_{n,m} \log_2\left(1+\frac{h_{n,m,k}p_{n,m,k}}{N_0\omega_{n,m}}\right), \label{rate}
\end{equation}
where $h_{n,m,k}$ is the channel gain and $p_{n,m,k}$ is the transmit power of user $n$ in slot $k$. Note that the rate $r_{n,m,k}$ still varies across slots because the channel gain and transmit power can change. $N_0$ is the Gaussian noise power spectral density.

We define the transmission unit as a packet and one such packet is transmitted per time slot. This packet serves as a container for one or multiple feature maps, with $b_{n,m,k}$ indicating their exact number. This number is constrained by the channel capacity and the size of an individual feature map, which depends on the partition point $s_{n,m}$. Specifically, a feature map is represented as an $L^\text{h}_s \times L^\text{w}_s$ matrix, therefore the value of $b_{n,m,k}$ is then calculated as the total transmissible bits in a slot divided by the size of a single feature map:
\begin{equation}
    \setlength\abovedisplayskip{5pt}
    \setlength\belowdisplayskip{5pt}
    b_{n,m,k} = \left\lfloor \frac{r_{n,m,k}  t_\text{slot}}{D  L^\text{h}_s  L^\text{w}_s} \right\rfloor. \label{b}
\end{equation}

\subsection{Task-aware Progressive Packet Transmission Model} \label{ppt}
After local computation, the device uploads the intermediate feature maps generated by the local sub-model to the edge server. To avoid inter-device interference, we adopt frequency division multiple access (FDMA) for offloading. Furthermore, both small-scale Rayleigh fading and large-scale path loss are considered in the device–server uplink channel model.

We emphasize that in delay-critical AI inference tasks, the total latency budget is typically insufficient to allow transmission of all intermediate features. Moreover, for a given task type, the inference accuracy varies significantly across different input samples based on their intrinsic complexity. However, from the perspective of the device, it only has \textit{statistical} priors and cannot \textit{deterministically} model this per-sample complexity, and may waste resources on simple samples or fail to achieve sufficient accuracy on complex ones.

To address the challenges and adapt to per-sample complexity, we adopt a progressive packet transmission mechanism with server-side stopping control, as illustrated in Fig. \ref{fig:prog}. For each task, the device sequentially transmits the most informative feature map packets. At the end of each slot, the server evaluates the inference confidence based on the cumulative packets received thus far. The transmission process terminates once this confidence exceeds a predefined threshold.

\begin{figure}[!t]
    \centering
    \includegraphics[width=\columnwidth]{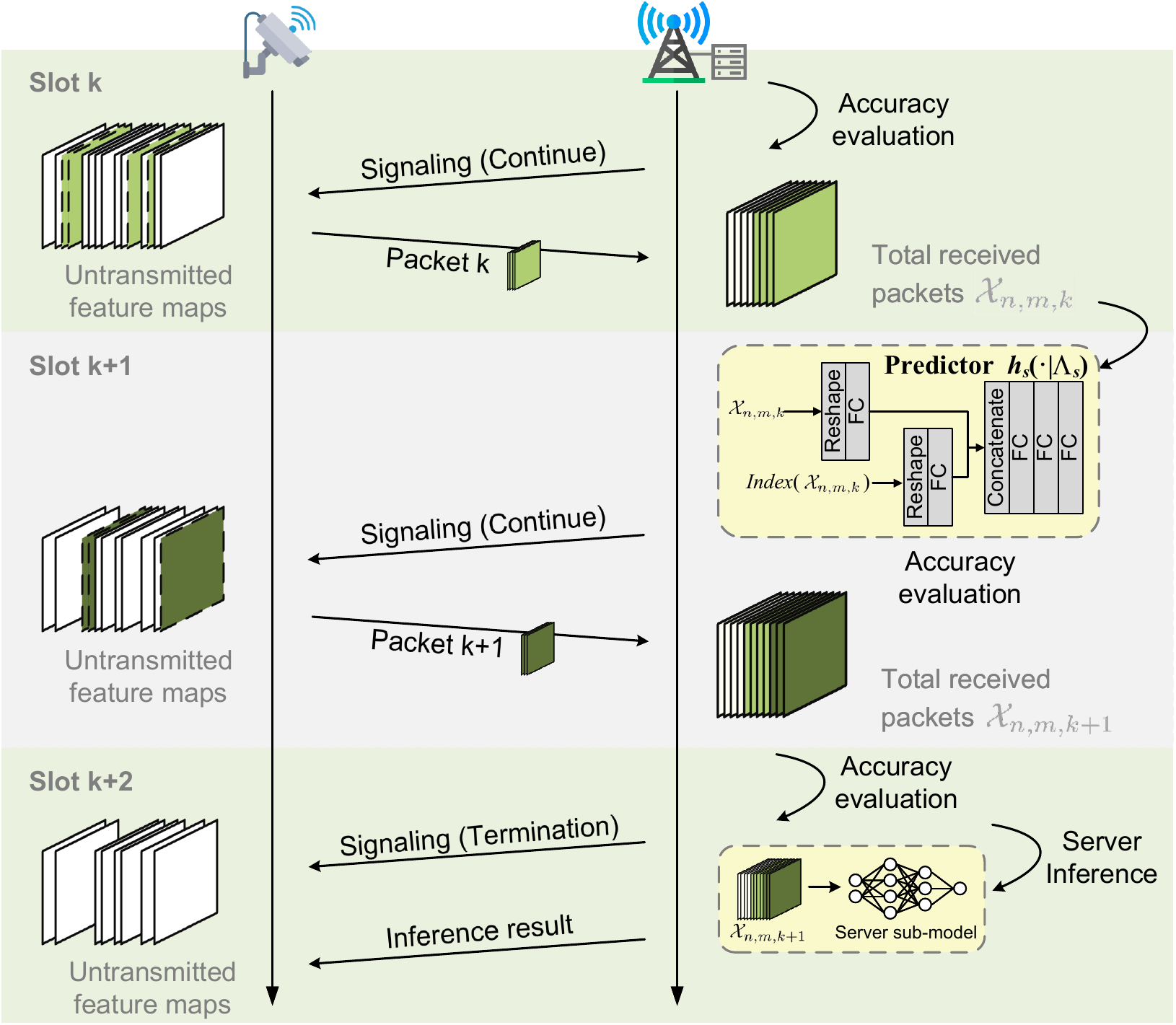} 
    \caption{Illustration of the progressive packet transmission mechanism.}
    \label{fig:prog}
\end{figure}

Following \cite{progressive}, we measure confidence using the predictive uncertainty of the inference. To this end, a lightweight MLP predictor $h_s(\mathcal{X}_{n,m,k}\,|\,\Lambda_s)$ is trained to estimate the uncertainty based on the feature map packets $\mathcal{X}_{n,m,k}$ received up to slot $k$, where $\Lambda_s$ denotes the model parameters. Its runtime is negligible compared to that of the main task inference model. 

The uncertainty is quantified by the predictive entropy:
\begin{equation} \label{eq:entropy}
H(\mathcal{X}_{n,m,k}) \triangleq -\sum_{l=1}^{L} Pr(l|\mathcal{X}_{n,m,k}) \log Pr(l|\mathcal{X}_{n,m,k}),
\end{equation}
where $L$ denotes the number of classes and $Pr(l|\mathcal{X}_{n,m,k})$ is the posterior probability of class $l$. The predictor $h_s$ is trained to approximate $H(\mathcal{X}_{n,m,k})$ by minimizing the loss between its output and the true entropy. The detailed implementation of this predictor will be discussed in Section~\ref{3c}.

We define $k_{n,m}$ as the total number of slots spent on the actual transmission, which corresponds to the frame duration excluding the local and edge computation times.
Accordingly, the communication energy consumption can be expressed as:
\begin{equation} \label{66}
    E_{n,m}^\text{tr}=\sum_{k=1}^{k_{n,m}} p_{n,m,k}t_\text{slot},
\end{equation}
where $p_{n,m,k}$ stands for the transmit power for device $n$ in frame $m$, slot $k$.

Since the edge server is generally powered by a stable power grid, the energy consumption of edge computation is not considered in this work. Therefore, the total energy consumption on frame $m$ for device $n$ is 
\begin{equation}
    E_{n,m} = E_{n,m}^\text{local}+E_{n,m}^\text{tr}= \alpha_n f_{n,m}^2 R_{s_{n,m}}^\text{local} + \sum_{k=1}^{k_{n,m}} p_{n,m,k}t_\text{slot}.
\end{equation}

\subsection{Batched Task Execution Model at the Edge} \label{batch}
Upon reception of feature map packets, the edge server, equipped with cached DNN models, executes the remaining layers of the inference task and may sends the inference results back to the user device. 




The edge server is assumed to have a fixed computing capability $f^\text{edge}_m$ per frame, and the corresponding edge computing delay given the edge computing workload $R_{s_{n,m}}^\text{edge}$ similar to that of the local computing part can be written as:
\begin{equation}
    t_{n,m}^\text{edge} = \frac {R_{s_{n,m}}^\text{edge}}{f^\text{edge}_m}. \label{6}
\end{equation}


To further enhance system efficiency and align with the periodic nature of task arrivals, the edge server adopts a lightweight batching mechanism. This mechanism synchronizes the execution of DNN inference across all $N$ devices within each frame $m$. The key principle is to ensure that edge-side inference for the batch starts only after all devices have completed their feature packets transmission, thereby avoiding server-side idle time and improving computational throughput.

This design is motivated by the common hard deadline $T$, which naturally imposes a shared time constraint for all tasks within the frame. The batch starting time, $t^\text{batch}_m$, is determined by the device with the largest edge inference delay, as this device dictates the latest possible moment the server can begin computation while still meeting all deadlines. Formally, $t^\text{batch}_m$ is expressed as:
\begin{equation} \label{eq:batch_time}
t^\text{batch}_m = t_m^\text{frame} + T - t_{\max}^\text{edge},
\end{equation}
where $t_m^\text{frame}$ denotes the starting time of the current frame, and $t_{\max}^\text{edge} = \max_{n \in \mathcal{N}} t_{n,m}^\text{edge}$ represents the maximum edge computation delay, as defined in (\ref{6}), among all devices, ensuring that the edge inference part of each task can be completed within the delay bound.

This start time $t^\text{batch}_m$ also defines the hard transmission deadline for each device. A task for device $n$ is considered successful if its feature transmission is completed before $t^\text{batch}_m$. Conversely, if additional time beyond $t^\text{batch}_m$ would be required for transmission, i.e., $t_{n,m}^\text{local} + t_{n,m}^\text{tr} > t^\text{batch}_m - t_m^\text{frame}$, 
the process is forcibly terminated. Consequently, the edge server performs inference using only the received features, inevitably resulting in degraded accuracy or wasted energy resources.

\subsection{Inference Task Accuracy Model}
For a given task $n$ in frame $m$ and its corresponding DNN partition point $s_{n,m}$, the inference accuracy depends mainly on the number of received feature maps $\mathcal{X}_{n,m,k}$.
According to \cite{29,23}, more  $\mathcal{X}_{n,m,k}$ generally leads to better inference accuracy, and the incremental gain diminishes as $\mathcal{X}_{n,m,k}$ increases. 
Therefore, by defining \(\mathrm{A}_{n,m}(s_{n,m},\mathcal{X}_{n,m,k})\) as a monotonically non-decreasing function, the average accuracy for all tasks in frame \(m\) can be expressed as:
\begin{equation}
    \mathrm{A}_m = \frac{1}{N}\sum_{n\in\mathcal{N}}\mathrm{A}_{n,m}(s_{n,m}, \mathcal{X}_{n,m,k}).
\end{equation}

\subsection{Problem Formulation} \label{opf}

We aim to maximize the long-term average inference accuracy of all devices, while ensuring that long-term average local energy consumption of each device remains below a prescribed threshold $\bar{E}_n$:
\begin{subequations}
    \begin{align}
    \mathcal{P}1:\max _{\boldsymbol{s}_m,\boldsymbol{\omega}_m,\boldsymbol{p}_m} &\,\, \frac{1}{M} \sum_{m=1}^M \mathrm{A}_{m}  \label{p1a}\\
    \text { s.t. }\,\,\,\, &\,\,\frac{1}{M} \sum_{m=1}^M E_{n,m} \leq \bar{E}_n,\forall n \in \mathcal{N}, \label{energyc}\\
    &\,\,\sum_{n=1}^{N}\omega_{n,m} \le \omega , \forall m \in \mathcal{M},  \label{cc}\\ 
    &\,\,~s_{n,m} \in\{0,1,2,\cdots,k_{m}\}, \label{dd} \\ 
    &\,\,~{p}_{n,m} \in (0,p_{\max}], \label{ee}
    \end{align}
\end{subequations}
where the optimization variables are the DNN partition point $\boldsymbol{s}_m=[s_{1,m},...,s_{N,m}]$, the bandwidth allocated to each device $\boldsymbol{\omega}_m=[\omega_{1,m},...,\omega_{N,m}]$ and the transmit power $\boldsymbol{p}_m=[p_{1,m},...,p_{N,m}]$. In the above problem, (\ref{p1a}) is the optimization goal of maximizing the long-term task inference accuracy, 
(\ref{energyc}) is the long-term energy constraint ensuring that the computing and transmission energy of each task remains within a stable range, 
and constraints (\ref{cc})-(\ref{ee}) are used to constrain the range of the optimization variables, where $\omega$ is the total bandwidth allocated to the system in a frame. Given the complexity of this mixed-integer nonlinear problem, we omit adaptive CPU frequency scaling and concentrate on optimizing the communication resources. Recall that in system model we have assumed that feature transmission time is predefined to be the remaining part of task deadline excluding local and edge computing delays, so here we no longer list the delay limit requirements.

Problem $\mathcal{P}1$ is essentially a \textit{stochastic optimization problem}. Its intractability arises from several coupled challenges: 
(1) \textit{The problem is inherently stochastic}: the system aims to optimize long-term inference accuracy and energy efficiency, but at the beginning of each frame only the instantaneous channel state is observable, while future channel conditions are unknown and difficult to predict, making direct offline optimization extremely challenging.
(2) \textit{The inference accuracy function is non-analytical}: although accuracy is known to be non-decreasing with respect to the number of received feature maps, its exact deterministic expression is unavailable, preventing direct use in analytical optimization.
(3) \textit{Some decision variables are integer-constrained}: for example, the selection of partition points $s_{n,m}$ is inherently discrete. When combined with the discrete local computational costs determined by the DNN architecture, this results in a non-convex feasible region and combinatorial complexity.
Overall, problem $\mathcal{P}1$ is a non-convex mixed-integer stochastic
optimization problem, which is generally complex to solve.

\section{ENACHI: Energy-Accuracy Hierarchical Optimization Algorithm for Split Inference}

In this section, we propose ENACHI to solve Problem $\mathcal{P}1$. The framework is built upon a \textit{nested drift-plus-penalty} architecture that operates at two granularities. At the task-level, an outer loop transforms the long-term stochastic problem into a per-frame optimization. To address the non-analytical accuracy objective, this loop employs a tractable surrogate model derived from statistical fitting. At the packet-level, a nested inner loop is designed to mitigate sample fluctuations, which works in conjunction with progressive packet transmission to dynamically adjust per-slot transmission policies, ensuring adherence to the  reference budget set at the task-level.

\subsection{Problem Conversion}

To make the stochastic optimization problem tracable, this section first decomposes the original problem at the task level by investigating the long-term stochastic dynamics using Lyapunov optimization. This approach establishes virtual queues that capture long-term constraints and decomposes them into per-frame decisions, thereby decoupling the problem from a temporal perspective. Under the drift-plus-penalty framework \cite{40}\cite{41}, the system stabilizes the virtual queues while minimizing the penalty, converting the long-term stochastic optimization into an online decision process that ensures system stability and sustained performance improvement.

We first create a virtual energy queue $Q_{n,m} \geq 0$ for each device $n$ to track the long-term energy constraint (\ref{energyc}). The actual energy consumption $E_{n,m}$ is treated as the per-frame \textit{arrival} to the queue, while the energy budget $\bar{E}_n$ is the constant \textit{service} rate. The queue backlog $Q_{n,m}$ thus represents the cumulative energy deficit relative to the budget. The queue dynamics are expressed as:
\begin{equation}
    Q_{n,m+1} = \left[Q_{n,m}+(E_{n,m} - \bar{E}_n)\right]^{+},
\end{equation}
where $[x]^+$ is defined as $\max\{x,0\}$, and the queue is initialized with $Q_{n,1}=0$. The queue accumulates the excess energy consumption whenever the per-frame energy usage $E_{n,m}$ exceeds the threshold $\bar{E}_n$, and decreases otherwise. We denote the queue vector at frame $m$ as of all users by $\boldsymbol{Q}_m = [Q_{1,m}, \ldots, Q_{N,m}]$.





\begin{lemma}
Using Lyapunov drift-plus-penalty framework, the long-term stochastic optimization problem $\mathcal{P}1$ can be transformed into the online, per-task optimization problem $\mathcal{P}1.1$:
\begin{align}
    \mathcal{P}1.1:\max _{\boldsymbol{s}_m,\boldsymbol{\omega}_m,\boldsymbol{p}_m} &\quad V \times \mathrm{A}_{m} - \sum_{n=1}^N Q_{n,m}E_{n,m}  \nonumber\\
    \textnormal { s.t. } \,\,\,\,&\quad\sum_{n=1}^{N}\omega_{n,m} \leq \omega , \forall n \in \mathcal{N}, \nonumber\\
    &\quad~s_{n,m} \in\{0,1,2,\cdots,k_{m}\}, \nonumber\\
    &\quad~{p}_{n,m} \in (0,p_{\max}].
\end{align}
\end{lemma}

\begin{proof}
    See Appendix~A.
\end{proof}



The objective of problem $\mathcal{P}1.1$ is twofold: the first term represents the optimization goal (accuracy), while the second term, the drift, serves to reduce queue backlog and maintain energy stability. The Lyapunov control parameter $V$ is a weighting coefficient that balances this trade-off. A larger $V$ prioritizes inference accuracy, potentially at the cost of higher transient energy consumption, whereas a smaller $V$ enforces stricter energy savings, possibly at the expense of accuracy.



However, problem $\mathcal{P}1.1$ still remains intractable. This stems from its mixed-integer non-linear programming (MINLP) nature, the non-analytical accuracy objective, and the inherent scheduling granularity mismatch between the task-level problem and packet-level variables. To efficiently solve $\mathcal{P}1.1$, we employ the hierarchical approach previously introduced, which is decomposed into two stages as follows:

\subsubsection{Stage I — Task-Level Resource Scheduling via Surrogate Model}
First, to address the non-analytical accuracy objective, we fit a tractable surrogate model $\widehat{\mathrm{A}}_{n,m}$ to empirical data obtained from a ResNet-50 model\cite{resnet} trained on ImageNet dataset. Substituting $\widehat{\mathrm{A}}_{n,m}$ into the task-level drift-plus-penalty problem $\mathcal{P}1.1$ yields a tractable per-frame optimization. Solving this determines the task-level decisions: the DNN partition point $\boldsymbol{s}_m$, bandwidth allocation $\boldsymbol{\omega}_m$, and estimated reference transmit power $\tilde{\boldsymbol{p}}_m$. This process also implicitly governs the batching deadline $t^\text{batch}_m$ as defined in Section \ref{batch}, enabling integrated control of computation and transmission.

\subsubsection{Stage II — Packet-Level Reference-Tracking and Adaptation}  
Then, the packet-level inner control loop addresses the granularity mismatch. It implements a reference-tracking policy that dynamically adjusts the per-slot transmit power $\boldsymbol{p}_m$ to follow the reference $\tilde{\boldsymbol{p}}_m$ from Stage I, adapting to fine-grained channel dynamics. This is integrated with an uncertainty-aware progressive feature packet transmission mechanism. The system computes predictive entropy to determine a stopping criterion, ensuring per-sample accuracy and mitigating the statistical gap of the surrogate model.

\subsection{Task-Level Resource Scheduling via Surrogate Model}

\begin{figure}[!t]
    \centering
    \includegraphics[width=0.96\columnwidth]{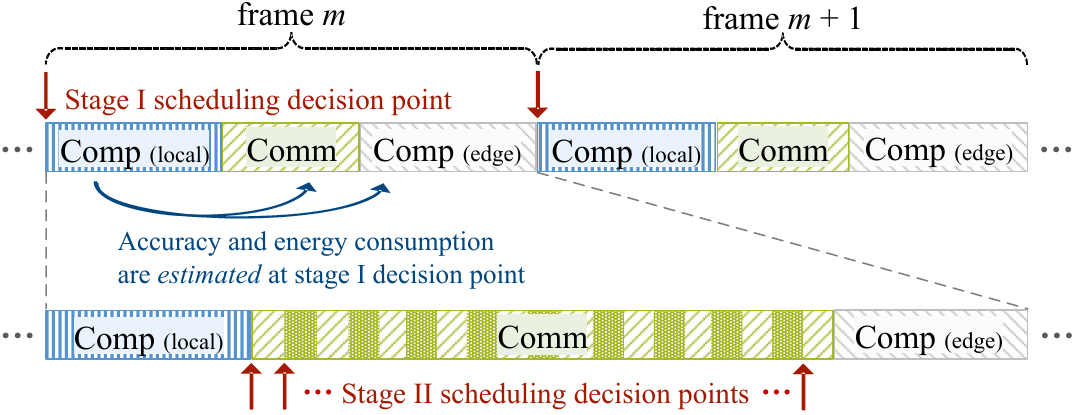} 
    \caption{Illustration of the two-stage scheduling workflow.}
    \vspace{-0.4cm}
    \label{fig:tt}
\end{figure}

We note that task-level optimization operates on a per-frame basis, which cannot precisely control per-slot power dynamics. Therefore, we first solve for split point $\boldsymbol{s}_{m}$, bandwidth $\boldsymbol{\omega}_{m}$ and a \textit{task-level reference transmit power}, denoted by $\tilde{\boldsymbol{p}}_m$, which serves as an estimated power budget for the entire frame, providing a baseline that guides the fine-grained, slot-wise adjustments in Stage II.

Using empirical accuracy curves obtained from ImageNet, we fit a tractable surrogate $\widehat{\mathrm{A}}_{n,m}(s_{n,m},\beta_{n,m})$ to represent the original accuracy function $\mathrm{A}_{n,m}(s_{n,m},\mathcal{X}_{n,m,k})$, where $\beta_{n,m} \in [0,1]$ denotes the proportion of input feature maps relative to the total feature maps at the corresponding split layer. Considering $\widehat{\mathrm{A}}_{n,m}(s_{n,m}, \beta_{n,m})$ is monotonically non-decreasing in $\beta_{n,m}$ and shows diminishing returns as $\beta_{n,m}$ increases, this surrogate allows the originally intractable $\mathcal{P}1.1$ to be reformulated into a per-task drift-plus-penalty problem.

\textbf{Surrogate Accuracy Model}: Based on our experimental results in Fig.~\ref{t1}, we model $\widehat{\mathrm{A}}_{n,m}$ as:
\begin{equation}
\setlength\abovedisplayskip{4pt}
\setlength\belowdisplayskip{4pt}
    \widehat{\mathrm{A}}_{n,m}(s_{n,m}, \beta_{n,m}) = -\frac{1}{a_{0}\beta_{n,m}-a_{1}}+a_{2}, \label{a}
\end{equation}
where $a_{0}\geq0$, $a_{1}\geq0$ and $a_{2}\geq0$ are constants depending on the DNN model architecture and partition layer position. 
Then, according to (\ref{rate}), the transmission proportion $\beta_{n,m}$ can be expressed as:
\begin{equation}
    \setlength\abovedisplayskip{5pt}
    \setlength\belowdisplayskip{5pt}
    \beta_{n,m} \!=\! \frac{b_{n,m}(s_{n,m})}{b_\text{total}(s_{n,m})} \!= \! \frac{\omega_{n,m} T_{n,m}^{\text{tr}} \log _2\!\left(1+\frac{h_{n,m} \tilde{p}_{n,m}}{\sigma^2}\right)}{ b_\text{total}(s_{n,m}) D L_s^\text{h} L_s^\text{w}}, \label{beta}
\end{equation}
where $h_{n,m}$ is the estimated average frame channel gain constant and $\sigma^2$ is the noisepower spectrum density, which in practical system implementations is commonly used to represent the overall noise power $N_{0}\omega_{n,m}$, and here we adopt the equivalent representation throughout the subsequent analysis and experiments for clarity and consistency.
$T_{n,m}^{\text{tr}}$ indicates the maximum allowable communication delay. As discussed in Section \ref{ppt}, the total latency budget is typically insufficient to allow transmission of all intermediate features.
Thus, the communication delay in each frame is:
\begin{equation}
    \setlength\abovedisplayskip{5pt}
    \setlength\belowdisplayskip{5pt}
    T_{n,m}^{\text{tr}}=T-(t_{n,m}^\text{local}+t_{n,m}^\text{edge}).
\end{equation}

Once the mathematical form of the accuracy function has been obtained, we can further reformulate problem $\mathcal{P}1.1$ as:
\begin{align}
    \setlength\abovedisplayskip{1.5ex}
    \setlength\belowdisplayskip{1.5ex}
    \mathcal{P}1.2:\!\!\!\max _{\boldsymbol{s}_m,\boldsymbol{\omega}_m,\boldsymbol{\tilde{p}}_m} &\,\, V \sum_{n=1}^N\widehat{\mathrm{A}}_{n,m}(s_{n,m},\beta_{n,m}) - \sum_{n=1}^N Q_{n,m}\tilde{E}_{n,m}  \nonumber\\
    \text { s.t. } \,\,\,\,&\,\, \text{(\ref{cc}), (\ref{dd}) and (\ref{ee}}),
\end{align}
where $\tilde{E}_{n,m}=E_{n,m}^\text{local}+\tilde{p}_{n,m}T_{n,m}^\text{tr}$ denotes the estimated energy consumption of (\ref{66}) associated with the estimated power $\tilde{p}_{n,m}$.

For notational simplicity, starting from $\mathcal{P}1.2$, we omit the frame index $m$ in this subsection since the scheduling and optimization are performed within the current frame.

To efficiently solve this MINLP problem, we seek to decouple the optimization of the resource variables $(\boldsymbol{\omega,\tilde{p}})$ and the DNN partition point $\boldsymbol{s}$. We then have the following lemmas for problem $\mathcal{P}1.2$.

\begin{lemma} \label{l2}
For a fixed DNN partition point $\boldsymbol{s}$, problem $\mathcal{P}1.2$ is concave with respect to $\tilde{\boldsymbol{p}}$ when $\boldsymbol{\omega}$ is fixed. The Karush–Kuhn–Tucker (KKT) conditions then yield the conditional optimal transmit power:
\begin{equation}
    \tilde{p}_{n}^*=
    \frac{\sigma^{2}}{h_{n}}
    \left(2^{\gamma}\exp\!\left( 2 W\!\left[\frac{1}{2\sqrt{2}^{\,\gamma}}
    \sqrt{
        \frac{
            \ln 2 \gamma h_{n} V  
        }{
            a_{1} \sigma^{2}  T_{n}^{\mathrm{tr}} Q_{n}
        }
    }\right] \right) -1 \right),
\end{equation}
where $W(\cdot)$ denotes the Lambert $W$ function, and $\gamma\!=\!\frac{a_{1}b_\text{total}(s_{n})DL_{s}^\text{L}L_{s}^\text{w}}{a_{0}\omega_{n}T_{n}^\text{tr}}$ for notational convenience.
\end{lemma}

\begin{proof}
    See Appendix~B.
\end{proof}

\begin{lemma} \label{l3}
Problem $\mathcal{P}1.2$ is a monotonically increasing function of $\boldsymbol{\omega}$ when $\tilde{\boldsymbol{p}}$ and $\boldsymbol{s}$ are fixed.
\end{lemma}
\begin{proof}
    This can be easily proved by deriving the derivative of $\mathcal{P}1.2$ with each $\boldsymbol{\omega}$, and thus it is omitted here.
\end{proof}



\newcommand{\SolveKKT}{\text{SolveKKT}}

\begin{algorithm}[t] 
    \caption{Task-Level Iterative Bandwidth and Power Allocation Algorithm}
    \label{suanfa1}
    \begin{algorithmic}[1]
        \STATE \textbf{Input:} $\boldsymbol{s}_m$, $\omega$, $p_{\max}$, $\epsilon_{\text{conv}}$, $I_{\max}$, $\boldsymbol{Q}_m$;
        \STATE \textbf{Output:} Optimal $\boldsymbol{\omega}^*_m$, $\tilde{\boldsymbol{p}}^*_m$, and Utility $U_{s_{n,m}}(\boldsymbol{\omega}^*,\tilde{\boldsymbol{p}}^*)$;
        
        \STATE \textbf{Initialize:} $i=0$, $\omega_{n,m}^{(0)} = \omega / N$, $\tilde{p}_{n,m}^{(0)} = p_{\max}$ for all $n$;
        \STATE Compute initial utility $U_{s_{n,m}}^{(0)}$ and reward $\Phi_{n,m}^{(0)}$.
        
        \STATE \textbf{repeat}
            \STATE $i = i + 1$;
            \STATE Compute $\Phi_\text{total} = \sum_{k=1}^{N} \Phi_{k,m}^{(i-1)}(\tilde{p}_{k,m}^{(i-1)})$;
            \STATE Update $\omega_{n,m}^{(i)}$ based on (\ref{daikuan}) for all $n$;
            
            \STATE Update $\tilde{p}_{n,m}^{(i)}$ based on Lemma \ref{l2} for all $n$;
            
            \STATE Compute $U_{s_{n,m}}^{(i)}$ using (\ref{utility});
            \STATE Update $\Phi_{n,m}^{(i)}$ according to (\ref{reward});
        \STATE \textbf{until} $|U_{s_{n,m}}^{(i)} - U_{s_{n,m}}^{(i-1)}| < \epsilon_{\text{conv}}$ or $i \ge I_{\max}$;
        
        \STATE $U_{s_{n,m}}(\boldsymbol{\omega}^*_m,\tilde{\boldsymbol{p}}^*_m)) = \sum_{n=1}^N U_{s_{n,m}}^{(i)}$;
        \STATE \textbf{return} $\boldsymbol{\omega}^*_m = \boldsymbol{\omega}^{(i)}_m$, $\tilde{\boldsymbol{p}}^*_m = \tilde{\boldsymbol{p}}^{(i)}_m$, $U_{s_{n,m}}(\boldsymbol{\omega}^*_m,\tilde{\boldsymbol{p}}^*_m)$.
    \end{algorithmic}
\end{algorithm}

Based on Lemma \ref{l2} and \ref{l3}, we develop an iterative algorithm, detailed in Algorithm \ref{suanfa1}, to solve the problem efficiently, which jointly optimizes $\boldsymbol{\omega}$ and $\tilde{\boldsymbol{p}}$. First, we define a utility function $U_{s_n}$ based on the optimization objective:
\begin{equation}
    \setlength\abovedisplayskip{5pt}
    \setlength\belowdisplayskip{5pt}
    U_{s_n}(\omega_n,\tilde{p}_n) = V \widehat{\mathrm{A}}_{n}(s_{n},\beta_{n}) - Q_{n}\tilde{E}_{n}. \label{utility}
\end{equation}

Then, a unit-bandwidth reward function aiming to quantify the achievable benefit per unit bandwidth is defined as:
\begin{equation}
    \setlength\abovedisplayskip{5pt}
    \setlength\belowdisplayskip{5pt}
    \Phi_{n}(\tilde{p}_{n}) = U_{s_n}(\tilde{p}_{n},\omega_0), \label{reward}
\end{equation}
where $\omega_0$ represents the unit bandwidth. Using this reward, the total bandwidth $\omega$ is allocated among devices at iteration $i\geq0$ proportional to their respective contributions:
\begin{equation}
    \setlength\abovedisplayskip{5pt}
    \setlength\belowdisplayskip{5pt}
    \omega_{n}^{(i)} = \frac{\Phi_{n}^{(i)}(\tilde{p}_{n})}{\sum_{j=1}^N \Phi_{j}^{(i)}(\tilde{p}_{j})} \omega.\label{daikuan}
\end{equation}

Given this allocation, the estimated transmit power $\tilde{p}$ can be updated by Lemma \ref{l2}, ensuring that the utility is maximized under the given bandwidth constraints. 

The algorithm alternates between updating the bandwidth allocation based on (\ref{daikuan}) and transmit power via Lemma \ref{l2}, iterating until the utility function converges within a threshold $\epsilon_{\text{conv}}$ or total number of iterations exceeds $I_{\max}$.

For initialization, the total bandwidth $\omega$ is uniformly allocated across all devices, which is expressed as
$\omega_{n}^{(0)} = \frac{\omega}{N}$,
and the initial transmit power of each device $\tilde{p}_{n}^{(0)} $ is set to its maximum admissible value $p_{\max}$. Based on these initial allocations, the utility function $U_{s_n}^{(0)}$ and the reward function $\Phi_{n}^{(0)}$ are then computed to serve as the starting point for the subsequent iterative procedure.

After obtaining the optimal bandwidth $\boldsymbol{\omega}^*$ and transmit power $\boldsymbol{\tilde{p}}^*$ for fixed partition points, we next determine the optimal partition points. We adopt a greedy approach by evaluating the utility for each candidate in the finite partition set $\mathcal{S}$ and select the $\boldsymbol{s}$ that yields the maximum utility. Specifically, for each device $n$, we evaluate the best utility $U_{s_n}(\omega^*_n,\tilde{p}^*_n)$ for each candidate $s_n \in \mathcal{S}$ using $\omega^*_n$ and $\tilde{p}^*_n$. The split point that maximizes this utility, $s_n^\star = \arg\max_{s \in \mathcal{S}} U_{s_n}(\omega^*_n,\tilde{p}^*_n)$, is then selected, and its associated $(\omega^\star_n, \tilde{p}^\star_n)$ are retained as the global optimal task-level resource allocation.

Furthermore, to enhance scalability and reduce computational complexity in dense-user scenarios, a regional aggregation strategy can be employed. This approach groups users with similar channel conditions in geographical locations, computing a representative resource allocation for the cluster rather than iterating for each user individually.

\begin{algorithm}[t!]
\caption{ENACHI Algorithm}
\label{alg:erachi}
\begin{algorithmic}[1]
\STATE \textbf{Input:} $\boldsymbol{Q}_m$, $V, \mathcal{S}, \omega, p_{\max}, \bar{E}_n$.
\STATE \textbf{Output:} Decisions $[\boldsymbol{\omega}^*_m, \boldsymbol{p}^*_m, \boldsymbol{s}^*_m]$

\STATE Initialize partition vector $\boldsymbol{s}^*_m = [s_0, \dots, s_0]$.
\STATE \textbf{for} each user $n = 1$ to $N$ \textbf{do}
    \STATE \quad \textbf{for} each candidate partition point $i \in \mathcal{S}$ \textbf{do}
        \STATE \quad \quad Calculate $U_i(s_{n,m}^*,\tilde{p}_{n,m}^*)$ using \textbf{Algorithm 1};
    \STATE $\boldsymbol{s}^\star_m = \arg\max_{i \in \mathcal{S}} \boldsymbol{U}_i$;

\STATE Solve $[\boldsymbol{\omega}^\star_m, \boldsymbol{\tilde{p}}^\star_m]$ using \textbf{Algorithm 1}$(\boldsymbol{s}_m^\star, \boldsymbol{Q}_m)$;

\STATE Initialize inner power queues $q_{n,m,1} = 0$ for all $n$.
\STATE Initialize received packets set $\mathcal{X}_{n,m,0} = \emptyset$ for all $n$.

\STATE \textbf{for} each time slot $k = 1, 2, \dots, K$ \textbf{do}
    \STATE \quad \textbf{for} each user $n = 1$ to $N$ \textbf{do}
        \STATE \quad \quad \textbf{if} user $n$ is not terminated \textbf{then}
            \STATE \quad \quad \quad Calculate $p_{n,m,k}^*$ using (\ref{calp});
            \STATE \quad \quad \quad \textbf{for} $\!j\!=\!1,2,\cdots\!,\min\{|\mathcal{X}_{n,m}|\!-\!|\mathcal{X}_{n,m,k}|,b_{n,m,k}\}$ \textbf{do}
            \STATE \quad \quad \quad \quad Select feature packet according to (\ref{eq:packet_selection});
            \STATE \quad \quad \quad Transmit feature packet $\Psi_{n,m,k}$ with $p_{n,m,k}^*$;
            \STATE \quad \quad \quad Updates $\mathcal{X}_{n,m,k} = \mathcal{X}_{n,m,k-1} \cup \Psi_{n,m,k}$;
            \STATE \quad \quad \quad Evaluates uncertainty $h_s(\mathcal{X}_{n,m,k}\,|\,\Lambda_s)$;
            \STATE  \quad \quad \quad \textbf{if} $h_s(\mathcal{X}_{n,m,k}|\Lambda_s)\! < \!H_\text{th}$ or deadline reached \textbf{then}
                \STATE \quad \quad \quad \quad Send TERMINATION signal to user $n$;
            \STATE \quad \quad \quad Update $q_{n,m,k}$.

\STATE Observe actual energy $E_{n,m}$ for all $n$;
\STATE Update $Q_{n,m+1} = [Q_{n,m} + E_{n,m} - \bar{E}_n]^+$.
\end{algorithmic}
\end{algorithm}

\subsection{Packet-Level Reference-Tracking and Adaptation} \label{3c}

While Stage I provides a frame-wise resource reference based on a tractable surrogate model, this solution is derived from statistical fitting, failing to capture the variability of individual inputs. 

To address the limitations of the task-level solution, we seek to enable the per-slot adaptation. Since the partition point $\boldsymbol{s}_m$ and bandwidth $\boldsymbol{\omega}_m$ have already been determined in the first stage, the online adaptation focuses on adjusting the transmit power $\boldsymbol{p}_{n,m,k}$ and regulating packet delivery via task-aware progressive transmission to balance per-slot efficiency with long-term energy sustainability.

Consequently, the task-aware packet-level inner control loop is introduced. Its objective is twofold: opportunistically maximize the per-slot transmission of feature packets based on instantaneous channel conditions, and ensure the long-term average transmit power robustly tracks the reference $\tilde{p}_{n,m}^\star$ from the task-level. This forms the per-slot optimization problem:
\begin{subequations}
    \begin{align}
        \mathcal{P}2.1:  \max _{{p_{n,m,k}}} &\quad \frac{1}{k_{n,m}}\sum_{k=1}^{k_{n,m}} b_{n,m,k} \\
        \text { s.t. } &\quad \frac{1}{k_{n,m}}\sum_{k=1}^{k_{n,m}}p_{n,m,k}\leq \tilde{p}_{n,m}^\star, \label{2longe}\\
        &\quad p_{n,m,k} \in (0,p_{\max}], \label{22c}
    \end{align}
\end{subequations}
where $b_{n,m,k}$ represents the maximal packet size defined in (\ref{b}), $p_{n,m,k}$ denotes the per-slot transmit power, and ${k_{n,m}}=\frac{T_{n,m}^\text{tr}}{t_\text{slot}}$ is the actual transmit slot number.

To enforce the long-term power constraint (\ref{2longe}) in the per-slot domain, we adopt a finite-horizon Lyapunov-based optimization for the inner control loop, as it operates over the limited number of slots within a single frame. A virtual power queue $q_{n,m,k}\geq0$ is defined for each device $n$ within frame $m$, which tracks the deviation between the accumulated transmit power $p_{n,m,k}$ and the reference power $\tilde{p}_{n,m}$ obtained from the task-level stage. This approach decomposes the problem into tractable, per-slot subproblems, analogous to the task-level optimization. We define the virtual power queue dynamics as:
\begin{equation}
    q_{n,m,k+1} = \left[ q_{n,m,k}+p_{n,m,k}-\tilde{p}_{n,m} \right]^+,
\end{equation}
where $q_{n,m,1}=0$. 


Then, we transform problem $\mathcal{P}2.1$ into a multiple deterministic, per-slot optimization problem, which aims to opportunistically maximize the drift-plus-penalty term. This transformed drift-plus-penalty problem can be stated as:
\begin{subequations}
    \begin{align}
        \mathcal{P}2.2:  \max _{{p_{n,m,k}}} &\quad v\times b_{n,m,k} -q_{n,m,k}p_{n,m,k}\\
        \text { s.t. } &\quad \text{(\ref{2longe}), (\ref{22c}),}
    \end{align}
\end{subequations}
where $v$ is the inner-loop Lyapunov control parameter. Analogous to $V$ in the task-level, $v$ governs the trade-off in the packet-level optimization, balancing the objective of maximizing per-slot transmitted packets $b_{n,m,k}$ against the need to stabilize the power queue $q_{n,m,k}$.

Problem $\mathcal{P}2.2$ is a strictly concave optimization problem with respect to the variable $p_{n,m,k}$. As detailed in the Appendix~C, this holds because the Hessian matrix of the objective function is symmetric and positive definite. Therefore, the optimal solution for the per-slot transmit power $p_{n,m,k}^\star$ can be derived analytically using the KKT conditions:
\begin{equation}
    p_{n,m,k}^\star = \frac{v\,\omega_{n,m}}{q_{n,m,k} D L_{s}^\text{h}L_{s}^\text{w}} \ln2 - \frac{\sigma^2}{h_{n,m,k}}, \label{calp}
\end{equation}
where $h_{n,m,k}$ is the deterministic channel gain of time slot $k$.

This solution yields a dynamic per-slot power policy that adapts to instantaneous queue states and channel dynamics. This per-slot policy thus translating the static, task-level budget into a fine-grained, adaptive control action.

Once the optimal per-slot power $p_{n,m,k}^\star$ is determined, device starts the progressive transmission of feature packets, which generally consists of one or multiple feature maps. 

Instead of transmitting arbitrary feature maps, the server-side controller performs importance-aware feature selection, choosing the \textit{most informative} and \textit{yet untransmitted} features at the beginning of each slot. Following the methodology in \cite{zhyx}, the importance of the $j$th parameter from the split layer of DNN model is quantitatively estimated through a first-order Taylor expansion of the learning loss $\mathcal{L}$, which can be obtained from the last round of model training as computed using the back-propagation algorithm: $\tilde{I}(w_j) = \left(\frac{\partial\mathcal{L}}{\partial w_{j}} \cdot w_{j}\right)^{2}$. 

Let $\mathcal{X}_{n,m}$ denote the set of total feature maps generated from the split point $s_{n,m}$, then the total importance of the $i$-th feature map $\mathcal{X}_i \in \mathcal{X}_{n,m}$ is the sum of these values for all parameters $w_j$ within the corresponding filter that generates it, represented by $g_c(\mathcal{X}_i) = \sum_{w_j \in  \mathcal{X}_i} \tilde{I}(w_j), \forall \mathcal{X}_i \in \mathcal{X}_{n,m}.$

Recall that $\mathcal{X}_{n,m,k}$ is the cumulative feature maps transmitted by slot $k$, and let $\Psi_{n,m,k}$ denote the feature packet transmitted at slot $k$. During each slot, as long as the packet size satisfies $|\Psi_{n,m,k}{}|<\min \{|\mathcal{X}_{n,m}|-|\mathcal{X}_{n,m,k}|,b_{n,m,k}\}$, new feature maps are iteratively added according to their importance scores as:
\begin{equation} \label{eq:packet_selection}
\Psi_{n,m,k} \!=\! \Psi_{n,m,k} \cup \left\{\underset{\!\!\mathcal{X}_i \in \left(\mathcal{X}_{n,m}\backslash\mathcal{X}_{n,m,k-1}\backslash\Psi_{n,m,k}  \right)}{\arg\max}  g_c(\mathcal{X}_i)\! \! \right\}.
\end{equation}

Upon receiving the packets at the end of each slot, the server performs an interim inference using the cumulative set of all packets received thus far. Then, the server evaluates the inference reliability using the trained uncertainty predictor $h_s\left(\mathcal{X}_{n,m,k}\,|\,\Lambda_s \right)$ mentioned in Section II.

The transmission process for device $n$ is terminated if the uncertainty condition $H_\text{th}$ is met, $h_s(\mathcal{X}_{n,m,k}|\Lambda_s) \!\!\le \!\!H_\text{th}$, or the hard transmission deadline, $t^\text{batch}_m$, is reached. Otherwise, the progressive transmission continues to the next slot.

\subsection{Performance Analysis}

We now characterize the theoretical performance of the proposed ENACHI algorithm by establishing a comparison with an ideal offline optimum of solving problem $\mathcal{P}1$. Let $\{\boldsymbol{s}_m^*, \boldsymbol{\omega}_m^*, \boldsymbol{p}_m^*\}$ represent the optimal offline decision sequence, and let $\mathrm{A}_m^*$ denote the corresponding maximum accuracy of frame $m$.

In our setting, decisions rely on the surrogate accuracy model $\widehat{\mathrm{A}}_m$ and the estimated energy consumption $\tilde{E}_{n,m}$ derived from the task-level scheduler. Define $\sum_{m=1}^M\widehat{\mathrm{A}}_m^\ddagger$ as the cumulative average inference accuracy of the proposed algorithm, which is obtained by solving $\mathcal{P}1.2$ in each frame.
We consider the wireless channel states to be temporally independent, without imposing any specific distributional assumption. The performance bound of the proposed algorithm is presented in the following theorem.

\begin{theorem}
Relative to the offline optimal solution, the cumulative inference accuracy of ENACHI is bounded by
\begin{equation}
    \sum_{m=1}^{M}\mathrm{A}_{m}^{\ddagger} \ge \sum_{m=1}^{M} \mathrm{A}_{m}^{*}-\frac{\theta_0M^2 + M(M-1)\delta_0\sum_{n=1}^N\theta_n}{V} - 2\xi_0,
\end{equation}
and the cumulative energy consumption violation satisfies
\begin{align}
    \sum_{m=1}^{M} &E_{n,m} \le M \bar{E}_n +\nonumber\\
    &\sqrt{2\theta_0M^2 +2M(M-1)\delta_0\sum_{n=1}^N\theta_n+4\xi_0V},
\end{align}
where $\theta_0\!\triangleq \!\sum_{n=1}^N \frac{1}{2}\theta_n^2$, $\theta_n \!\triangleq\! \max_m \{|E_{n,m}\!-\!\bar{E}_n| \}$, $\delta_0\!\triangleq \!\max_{\{n,m\}}\!\!\left\{\left|\tilde{E}_{n,m}\!-\!E_{n,m} \right|\right\}$ and $\xi_0\!\triangleq \!\max_{m}\left\{\left|\widehat{\mathrm{A}}_m\!-\!\mathrm{A}_{m} \right|\right\}$.
\end{theorem}
\begin{proof}
    We employ a two-step bounding technique to establish the proof; see Appendix~D for details.
\end{proof}
Theorem 1 indicates that the inference performance of the proposed algorithm can be bounded relative to its offline optimal counterpart. Meanwhile, the discrepancy between cumulative energy consumption of each device and its allocated budget is also limited. The worst-case performance can be enhanced by decreasing the upper bound of the energy usage deviation $\theta_n$, the maximum energy estimation error $\delta_0$ and the approximation error $\xi_0$. Additionally, the weight parameter $V$ allows balancing the trade-off between task inference performance and energy consumption of devices. In practical implementations, energy should be managed carefully to keep $\theta_n$, $\xi_0$ small, and $V$ should be chosen judiciously to achieve the best inference performance within the energy constraints.

\section{Experiments}\label{Section:Experiments} 
In this section, we conduct extensive simulations to evaluate the performance of our proposed ENACHI framework. We validate its effectiveness in maximizing long-term inference accuracy while ensuring device energy stability under strict deadlines. We validate our approach by comparing it against several state-of-the-art benchmark schemes.


\subsection{Simulation Setup} \label{sima}

\begin{figure}[!t]
    \centering
    \includegraphics[width=0.38\textwidth]{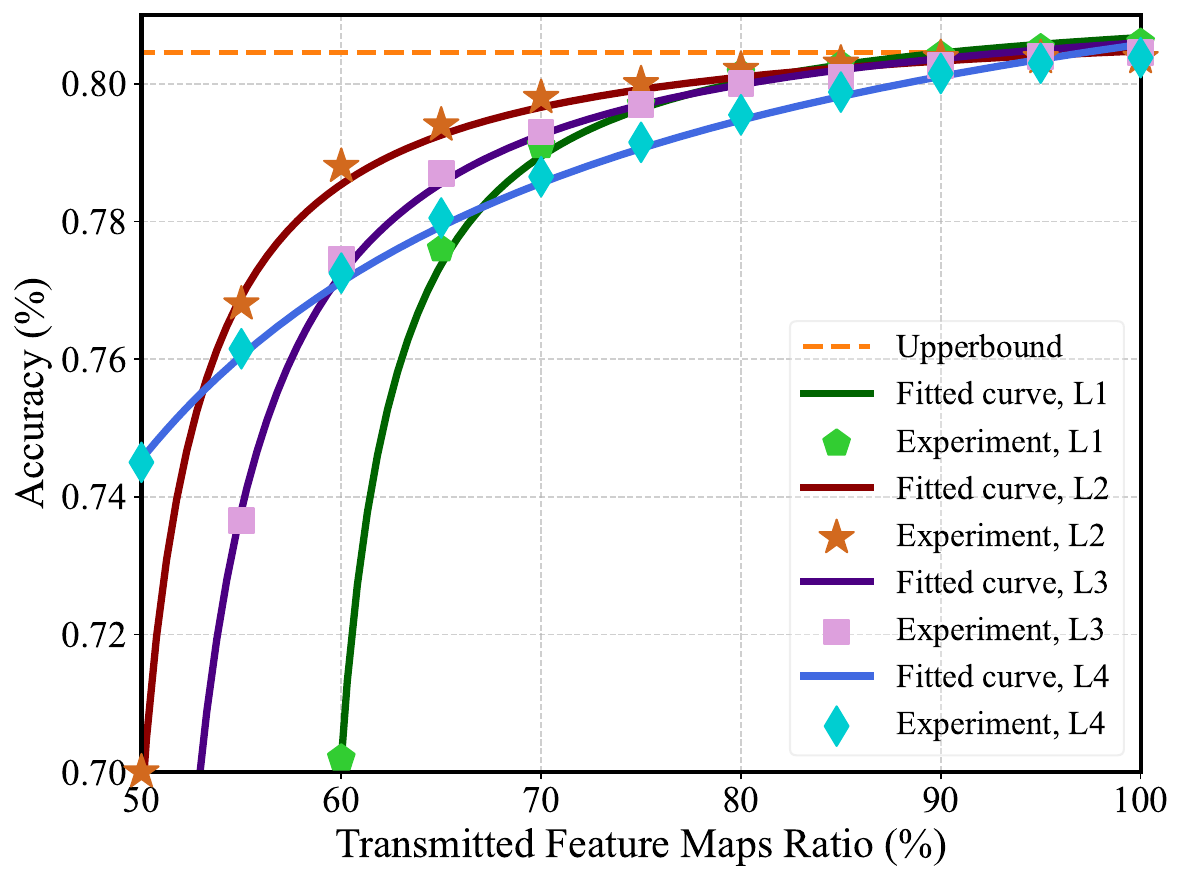}
    \vspace{-0.1cm}
    \hspace*{0.5cm}
    \caption{The experimental and fitted curves of image classification task. L1 to L4 are the selected representative partition layers from shallow to deep.}
    \label{t1}
\end{figure}

We simulate a split inference system with $N$ users and one edge server over a channel modeling both large-scale path loss and small-scale Rayleigh fading. The inference task employs a trained ResNet-50 model with upperbound accuracy of 80.38\%, and we use the ImageNet dataset. A set of feasible partition points $\mathcal{S}$ is predetermined. The accuracy–transmission relationship in (\ref{a}) is modeled by fitting to empirical data from the validation set of the dataset. All reported results are averaged over 1000 simulation rounds.

\begin{figure}[!t]
    \centering
    \includegraphics[width=0.8\columnwidth]{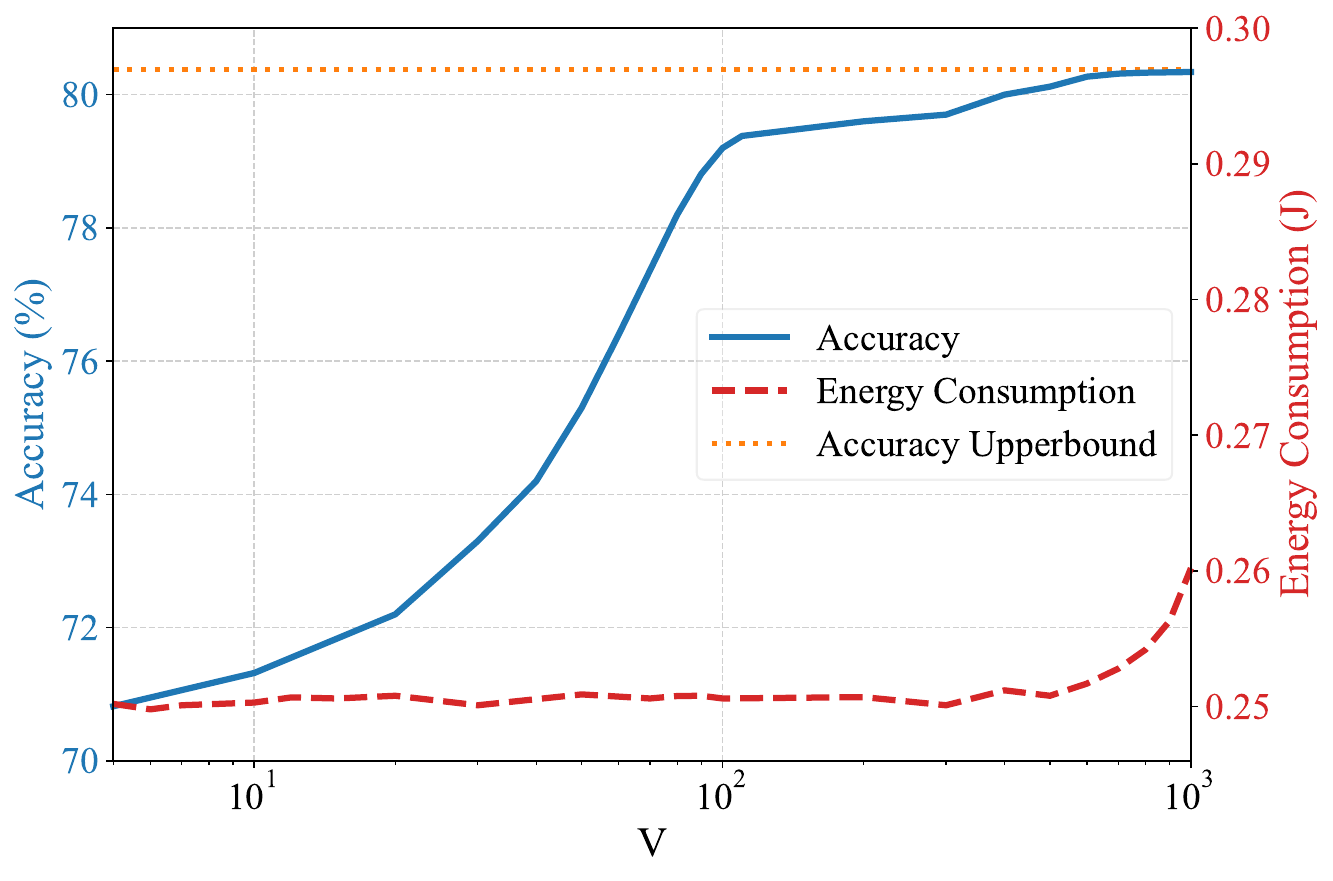} 
    \vspace{-0.3cm}
    \caption{Inference accuracy and energy consumption under different $V$.}
    \label{fig:v_tradeoff}
    \vspace{-0.5cm}
\end{figure}

\begin{figure*}[tbp!]
    \centering
    
    \subfloat[Accuracy under different deadlines.]{
        \includegraphics[width=0.305\textwidth]{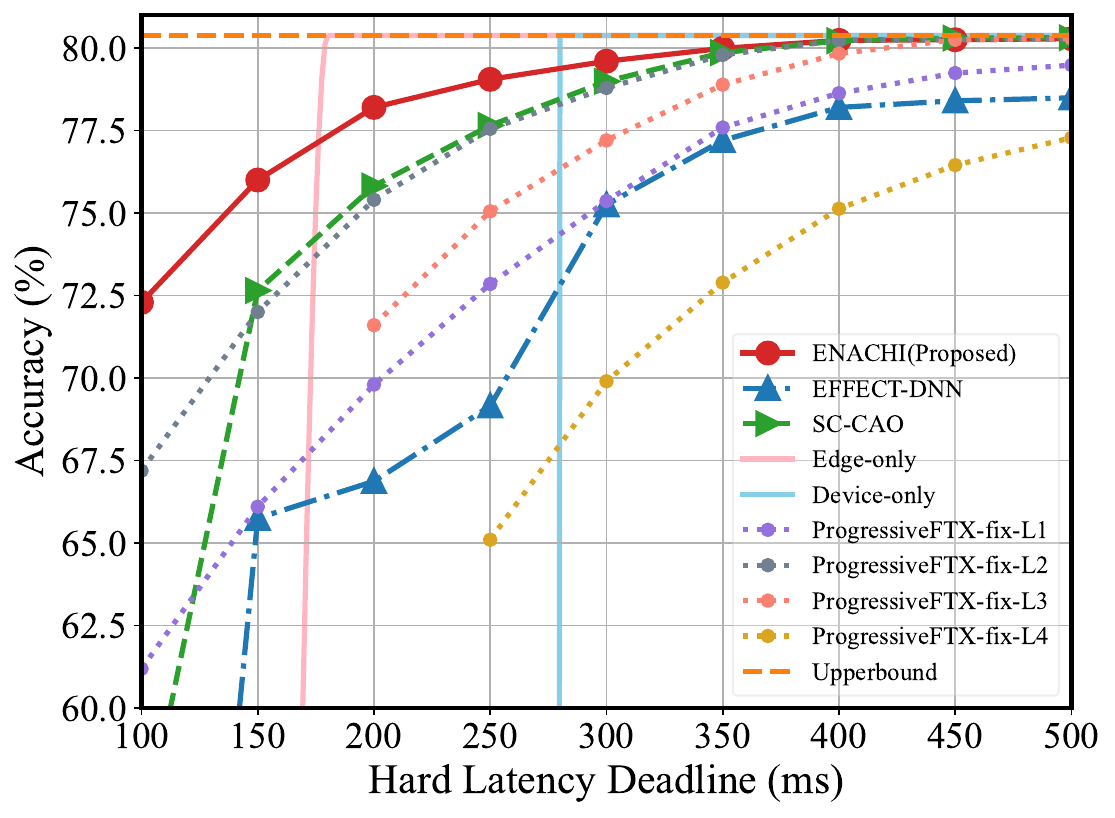}
        \label{fig:acc_ddl}
    }\hfill
    \setcounter{subfigure}{2}
    \subfloat[Accuracy under different bandwidth.]{
        \includegraphics[width=0.292\textwidth]{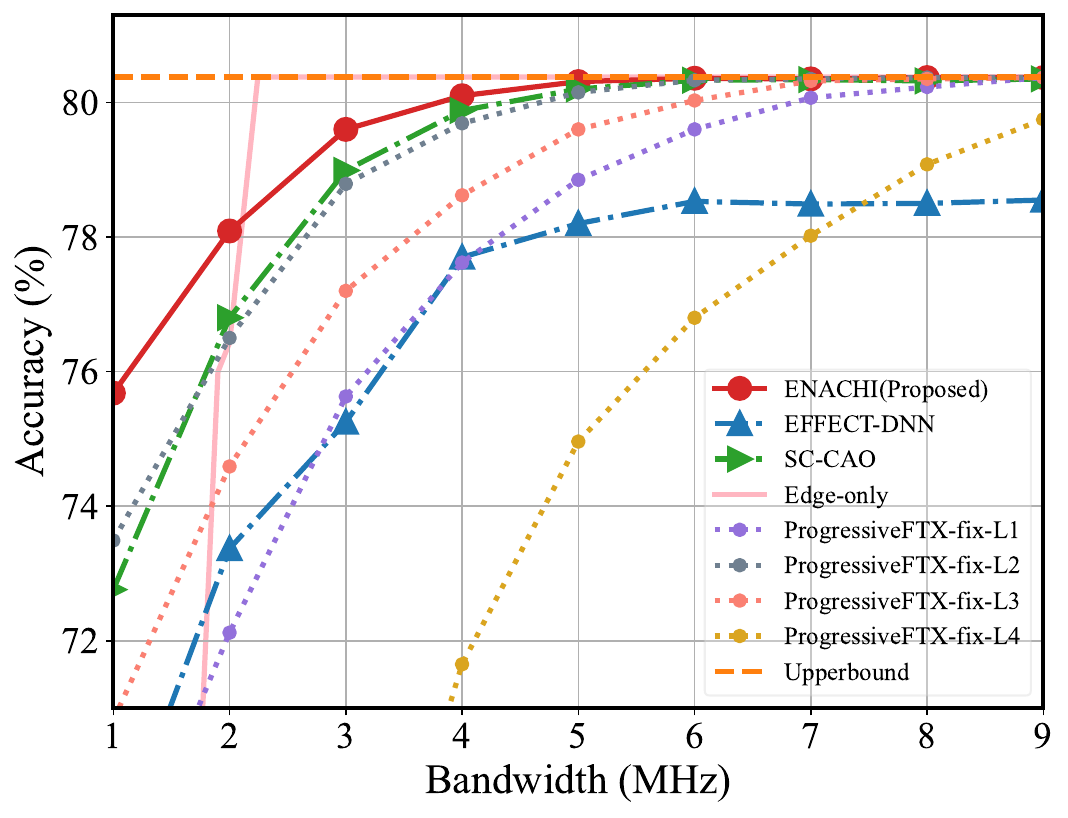}
        \label{fig:acc_bw}
    }\hfill
    \setcounter{subfigure}{4}
    \subfloat[Accuracy under different user numbers.]{
        \includegraphics[width=0.295\textwidth]{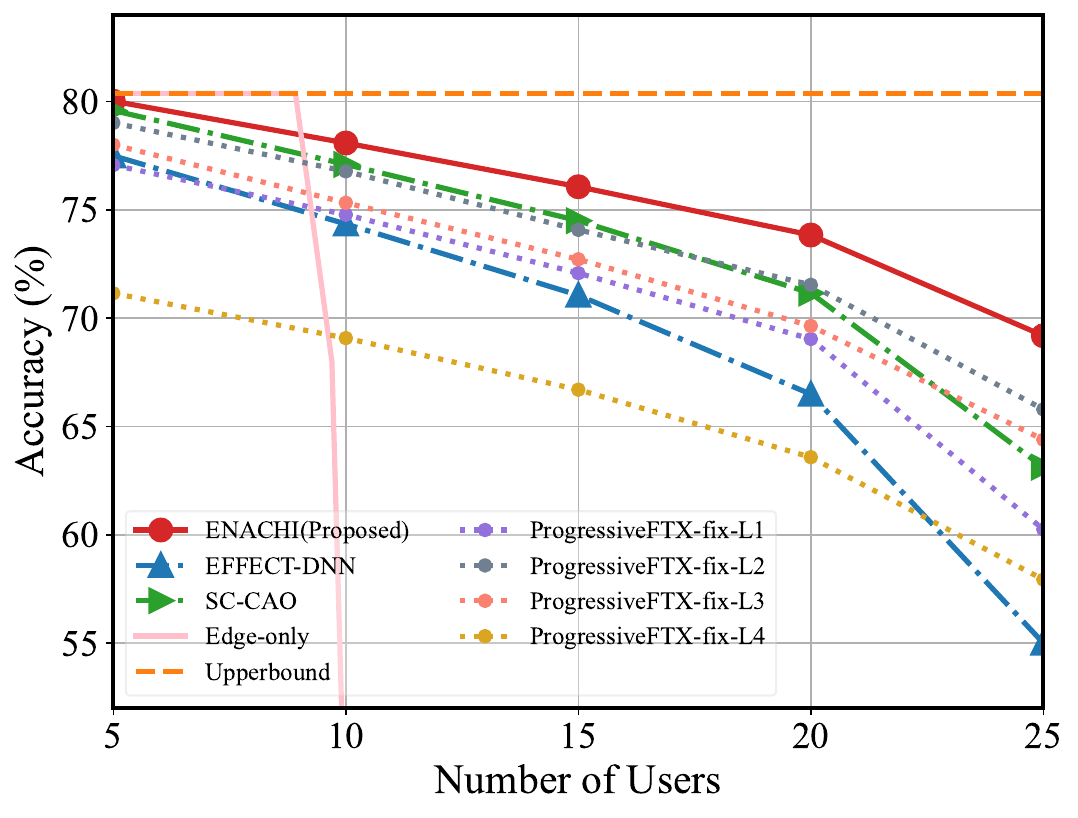}
        \label{fig:acc_users}
    }

    \vspace{-0.2mm}

    \setcounter{subfigure}{1}
    \subfloat[Energy consumption under different deadlines.]{
        \includegraphics[width=0.30\textwidth]{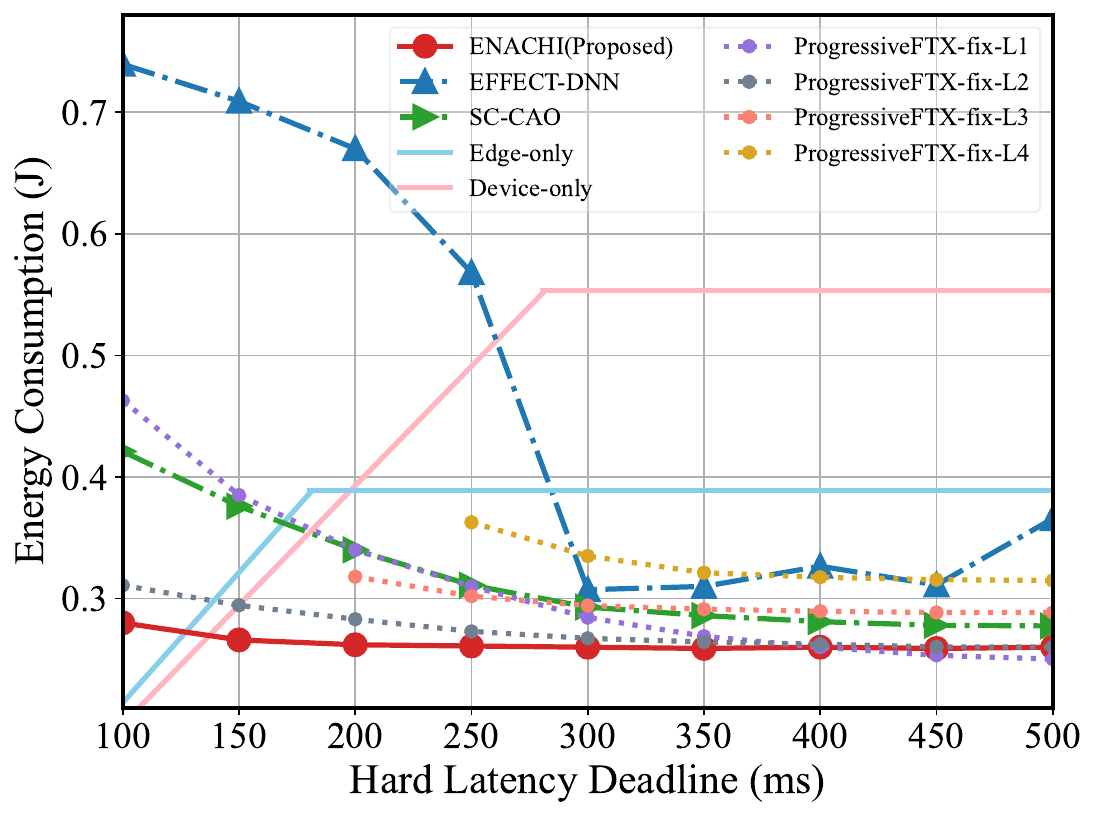}
        \label{fig:eng_ddl}
    }\hfill
    \setcounter{subfigure}{3}
    \subfloat[Energy consumption under different bandwidth.]{
        \includegraphics[width=0.30\textwidth]{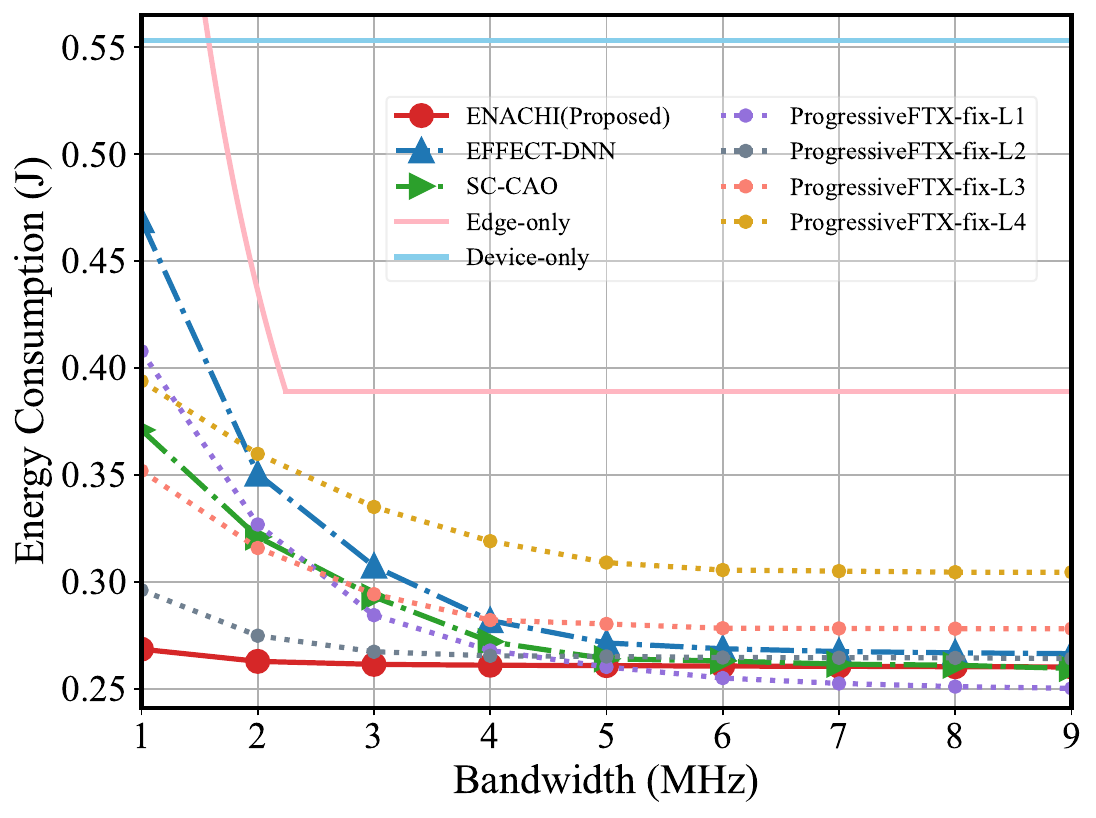}
        \label{fig:eng_bw}
    }\hfill
    \setcounter{subfigure}{5}
    \subfloat[Average energy consumption under different user numbers.]{
        \includegraphics[width=0.30\textwidth]{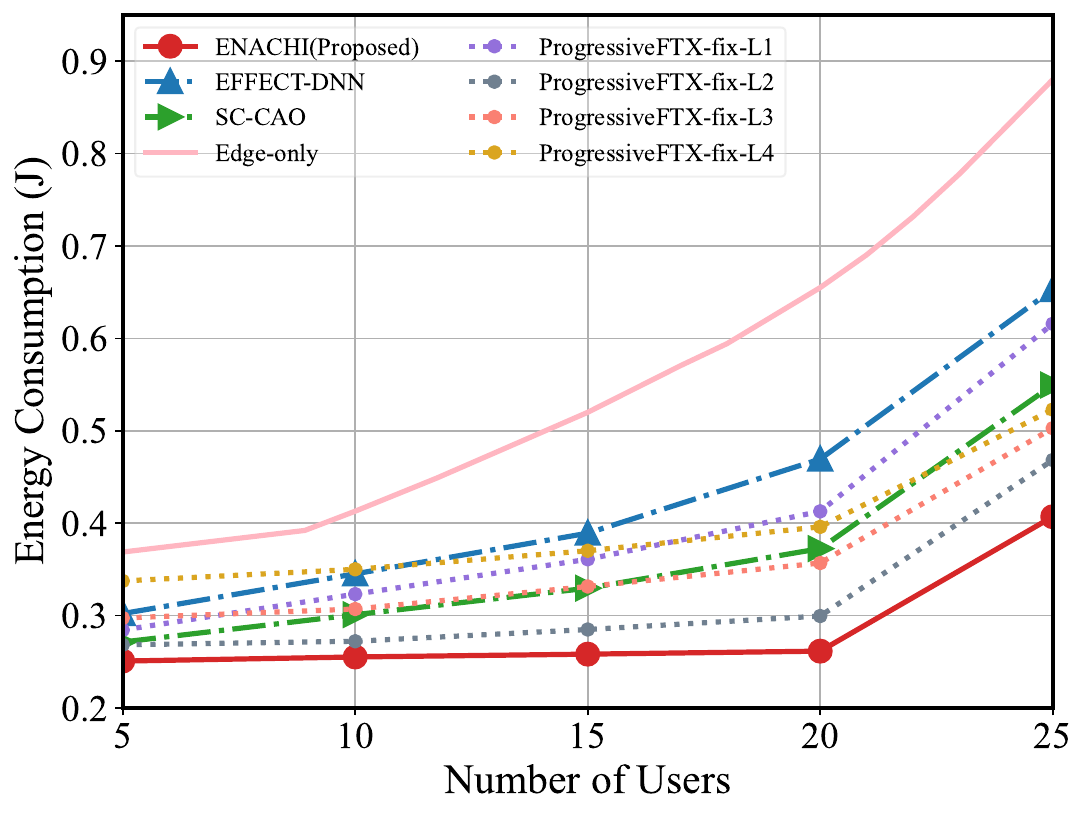}
        \label{fig:eng_users}
    }

    \caption{Performance of the proposed algorithm and benchmarks on ImageNet.}
    \label{fig:overall_results}
    \vspace{-0.4cm}
\end{figure*}

Unless otherwise specified, the main simulation parameters are set according to the configurations in \cite{exp1,c1,p3} and summarized in Table~\ref{tab:sim_params}.

\begin{table}[htbp]
    \centering
    \caption{Simulation Parameters}
    \label{tab:sim_params}
    \begin{tabular}{l|l}
        \hline
        \textbf{Parameter} & \textbf{Value} \\
        \hline
        Noise Power Spectral Density ($\sigma^2$) & 1 $\times 10^{-13}$ W \\
        Device CPU Clock Speed ($f_{n,m}$) & 2.0 GHz \\
        Edge Server GPU frequency ($f_m^\text{edge}$) & 20 GHz \\
        Device Chip Power Constant ($\alpha_n$) & $10^{-28}$ \\
        Maximum Transmit Power ($p_{\max}$) & 2 W \\
        Long-term Energy Budget ($\bar{E}_n$) & 0.25 J \\
        Outer Lyapunov Control Parameter ($V$) & 50 \\
        Inner Lyapunov Control Parameter ($v$) & 5 \\
        Frame length & 300ms \\
        Slot length  & 1ms \\
        \hline
    \end{tabular}
    \vspace{-0.2cm}
\end{table}

For comparison, the following benchmarks are considered:

(1) {EFFECT-DNN\textnormal{\cite{c2}}}: A Lyapunov-based framework that minimizes long-term device energy consumption. It targets \textit{average} latency bound rather than hard deadlines and lacks hierarchical task-aware scheduling. 

(2) {Semantic Communication based Computation-Aware Offloading (SC-CAO)\textnormal{\cite{c3}}}: This framework uses semantic communication to jointly optimize the compression ratio, computation, and transmission resources. It employs a myopic optimization strategy that focuses on the current frame rather than long-term system stability.

(3) {Progressive Feature Transmission (ProgressiveFTX)\textnormal{\cite{progressive}}}: This is an ablation study of our proposed algorithm where DNN split point is fixed. We test 4 representative partition points, from shallow to deep. 

(4) {Edge-Only}: Offloads the entire  task to the edge server. 

(5) {Device-Only}: Executes the entire task locally.

\subsection{Simulation Results}
\label{sec:performance_comparison}
Our experimental evaluations begin by validating the surrogate model and control parameters in single-user scenarios to establish a performance baseline and subsequently extend to multi-user scenarios to evaluate system scalability.

\subsubsection{Surrogate Model Validation}
We validated our hyperbolic surrogate model by selecting representative partition points L1 to L4 from the shallow to deep stages of ResNet-50. These points correspond to the 1st, 4th, 8th, and 14th convolutional layers respectively. As illustrated in Fig.~\ref{t1}, our fitted model exhibits a strong agreement with the empirical accuracy curves. This close match confirms that the surrogate function effectively captures the diminishing returns of accuracy and transmitted feature ratio, justifying its use in our framework. The experiments are conducted using the ImageNet dataset. 

\subsubsection{Impact of V} 
\label{sec:impact_of_v}
The parameter $V$ in the drift-plus-penalty formulation (17) is a crucial control knob that governs the trade-off between maximizing inference accuracy and maintaining energy queue stability. A higher $V$ prioritizes the accuracy objective, while a lower $V$ emphasizes strict adherence to the long-term energy budget. Fig.~\ref{fig:v_tradeoff} illustrates the impact of varying $V$ over several orders of magnitude on the long-term average accuracy and energy consumption of the system in a single-user scenario.
As depicted in Fig.~\ref{fig:v_tradeoff}, the behavior of the system changes significantly with the magnitude of $V$. For small values ($V \in [10^0, 10^1]$), the framework operates in an energy-conservative mode, maintaining low energy consumption at the cost of reduced accuracy. As $V$ increases into the range of $(10^1, 10^2]$, the system enters a balanced trade-off regime. Here, inference accuracy rises sharply with only a marginal increase in energy cost, indicating that our framework efficiently co-optimizes resources to achieve performance gains. When $V$ becomes very large ($V > 10^2$), the accuracy curve begins to saturate as it approaches its upper bound. In this regime, pursuing marginal accuracy gains incurs a large energy cost, signifying diminishing returns.

\subsubsection{Impact of frame deadline $T$}

We first evaluate performance under varying hard latency deadlines with a fixed 3~MHz bandwidth. Inference accuracy and device energy consumption are used as the key performance metrics. As shown in \figurename~\ref{fig:overall_results}\subref{fig:acc_ddl}, our framework achieves a remarkable 43.12\% accuracy gain over benchmarks at the stringent 100~ms deadline, reaching 72.5\% accuracy and smoothly rising to 80.1\% at 300~ms. Concurrently, it reduces energy consumption by up to 62.13\% as depicted in \figurename~\ref{fig:overall_results}\subref{fig:eng_ddl}, maintaining stable usage below 0.28~J. This performance is attributed to the dynamic algorithm with energy reference, resulting in stable energy consumption while achieving robust accuracy against varying deadlines. In comparison, SC-CAO also considers hard deadline and show robustness to changes in deadline. ProgressiveFTX and Device-only schemes are inflexible, becoming entirely infeasible for deadlines below 275~ms. The EFFECT-DNN framework, which targets average latency rather than hard deadlines, incurs an accuracy gap even under relaxed deadlines, and exhibits a significant energy performance degradation in scenarios with stringent deadline constraints. 

\subsubsection{Impact of user bandwidth $\omega$}

We evaluate the performance under varying channel bandwidth, with the hard deadline fixed at 300~ms. As shown in \figurename~\ref{fig:overall_results}\subref{fig:acc_bw}, our proposed framework achieves the best accuracy-bandwidth trade-off across all conditions. The advantage is most pronounced in the communication-constrained region of 1 to 3~MHz. Specifically, at 1~MHz, our method realizes a 9.39\% accuracy gain over benchmarks to reach 76\% accuracy, eventually saturating near 6~MHz. Regarding energy stability in \figurename~\ref{fig:overall_results}\subref{fig:eng_bw}, our approach reduces energy consumption by 42.74\% at 1~MHz compared to baselines and maintains the lowest usage across all bandwidths.
This robustness under low bandwidth is attributed to the progressive transmission mechanism. In communication-constrained scenarios, this mechanism selects high-value features based on their importance, which significantly saves communication overhead.
In comparison, 
the SC-CAO framework, which also considers data compression and dynamic bandwidth allocation, shows some robustness in communication-constrained scenarios. 
The EFFECT-DNN framework suffers from significant accuracy degradation and high energy consumption at low bandwidths, mainly due to its lack of a packet-level dynamic transmission mechanism.
The Edge-Only scheme is entirely infeasible below 2.5~MHz, and ProgressiveFTX schemes perform poorly at some points due to the inflexible strategy.

\subsubsection{Scalability Analysis}

Finally, we evaluate the scalability of the system in a multi-user scenario. In this experiment, we fix the total system bandwidth to 20~MHz and increase the number of users from 5 to 25, simulating an environment with increasing resource contention. As shown in \figurename~\ref{fig:overall_results}\subref{fig:acc_users}, as the per-user available bandwidth decreases, the accuracy of all schemes inevitably degrades. While our framework drops to nearly 70\% accuracy in bandwidth-scarce scenarios, it degrades gracefully and achieves a 14.19\% accuracy gain at 25 users compared to benchmarks. Most critically, \figurename~\ref{fig:overall_results}\subref{fig:eng_users} demonstrates strong energy scalability. As the user count rises, the energy cost remains remarkably flat below 0.28~J, realizing a 37.65\% energy saving at 25 users.
This stability is attributed to our two-tier Lyapunov framework which stabilizes individual energy consumption and mitigates inter-user competition through dynamic allocation.
In comparison, the SC-CAO  framework, despite also using dynamic allocation, suffers from linearly increasing energy consumption due to its myopic optimization.
The accuracy and energy performance of the EFFECT-DNN framework degrade once the user count exceeds 10, as the bandwidth available per user drops.
The Edge-Only and static-partition schemes demonstrate a lack of adaptability in congested networks.

\section{Conclusion}\label{Section:Conclusion}
In this paper, we aim to realize energy-aware, deadline-critical DNN offloading in multi-user split inference systems. The main challenges include the scheduling granularity mismatch between task-level and packet-level operations, and the lack of task-aware adaptation to the inherent heterogeneity of inference tasks. We proposed ENACHI, a novel ENergy-ACcuracy Hierarchical optimization for split-Inference framework built on a nested drift-plus-penalty architecture. ENACHI operates at two scales: a task-level outer loop uses a surrogate model to manage the long-term energy-accuracy trade-off and set a power reference, while a packet-level inner loop implements a reference-tracking policy to dynamically adjust per-slot packet transmission. This is integrated with an importance-aware progressive transmission and an adaptive uncertainty-based stopping criterion. Comprehensive simulations on ImageNet demonstrated that ENACHI significantly outperforms state-of-the-art benchmarks, achieving a better energy-accuracy trade-off and higher scalability, particularly under stringent deadlines.

\bibliographystyle{IEEEtran}
\bibliography{IEEEabrv,reference}

\onecolumn

\setcounter{equation}{28}

\appendices
\section{Proof of Lemma 1} \label{app1}
To characterize the congestion level of the virtual queues, we define a queue backlog vector $\Theta_{m}=(Q_{1,m}, Q_{2,m}, \cdots, Q_{n,m})$. For $\Theta_{m}$, we introduce the quadratic Lyapunov function, which is formally defined as:
\begin{equation}
	L(\Theta_{m}) = \frac{1}{2} \sum_{n=1}^{N} (Q_{n,m})^2,
\end{equation}
the size of which can provide a direct measure of the relative severity of queues accumulation. Clearly, the Lyapunov function is non-negative and we define $L(\Theta_{0})=0$.

Then, we define the Lyapunov drift function as
\begin{equation}
	\Delta(\Theta_{m}) = \mathbb{E}\big[ L(\Theta_{m+1}) - L(\Theta_{m}) \,\big|\, \Theta_{m} \big],
\end{equation}
which characterizes the expected change of the Lyapunov function between consecutive frames, i.e., $m$ and $m+1$, and it is crucial to the system stability. 

According to (12):
\begin{enumerate}
	\item if $Q_{n,m}\!+\!E_{n,m}\! -\! \bar{E}_n\! \geq \!0$, then $Q_{n,m+1}=Q_{n,m}+E_{n,m} - \bar{E}_n$, so $(Q_{n,m+1})^2=(Q_{n,m}+E_{n,m} - \bar{E}_n)^2$, $\forall n \in \mathcal{N}$;
	\item if $Q_{n,m}\!+\!E_{n,m} \!-\! \bar{E}_n \!\leq \!0$, then $Q_{n,m+1}\!=\!0\!>\!Q_{n,m}\!+\!E_{n,m} \!-\! \bar{E}_n$, so $(Q_{n,m+1})^2\!=\!0\!<\!(Q_{n,m}\!+\!E_{n,m} \!-\! \bar{E}_n)^2$, $\forall n \!\in \!\mathcal{N}$.
\end{enumerate}

In conclusion, we have 
\begin{equation}
	(Q_{n,m+1})^2 \leq (Q_{n,m}+E_{n,m} - \bar{E}_n)^2.
\end{equation}

By expanding the above equation, we can obtain that
\begin{align}
	\frac{1}{2} \sum_{n=1}^N (Q_{n,m+1})^2 &\leq \frac{1}{2}\sum_{n=1}^N(Q_{n,m})^2 +\frac{1}{2}\sum_{n=1}^N(E_{n,m} - \bar{E}_n)^2  +\sum_{n=1}^NQ_{n,m}\times(E_{n,m} - \bar{E}_n). \nonumber
\end{align}

Rearranging the terms in the above inequality and taking the expectation on both sides, we obtain
\begin{align}
	\Delta(\Theta_m) &= \frac{1}{2} \sum_{n=1}^N\mathbb{E}[(Q_{n,m+1})^2 \,\big|\, \Theta_{m}] - \frac{1}{2} \sum_{n=1}^N\mathbb{E}[(Q_{n,m})^2 \, \big|\, \Theta_{m} ] \nonumber\\ 
	&\leq \frac{1}{2} \sum_{n=1}^N\mathbb{E}[(E_{n,m} - \bar{E}_n)^2\,\big|\,\Theta_{m}] + \sum_{n=1}^N\mathbb{E}[Q_{n,m}(E_{n,m}-\bar{E}_n)\,\big|\,\Theta_{m}] \nonumber \\
	&\leq \theta_0 + \sum_{n=1}^N\mathbb{E}[Q_{n,m}(E_{n,m}-\bar{E}_n)\,\big|\,\Theta_{m}]
\end{align}
where $\theta_0\triangleq \sum_{n=1}^N \frac{1}{2}\theta_n^2$ and $\theta_n \triangleq \max_m \{|E_{n,m}-\bar{E}_n| \}$.  
Therefore, for a given Lyapunov function $L(\Theta_{m})$, when $E[L(\Theta_{m})]\leq\infty$, the upper bound of the subsequent Lyapunov drift can be expressed as
\begin{align}
	\Delta(\Theta_{m}) &= \mathbb{E}\big[ L(\Theta_{m+1}) - L(\Theta_{m}) \,\big|\, \Theta_{m} \big] \nonumber \\
	&\leq \theta_0 + \sum_{n=1}^N Q_{n,m} \times (E_{n,m}-\bar{E}_n),
\end{align}

The established upper bound of the Lyapunov drift quantifies the impact of current decisions on the long-term constraint, providing a basis for per-frame optimization.

To incorporate the original system objective, we introduce the penalty term in the Lyapunov optimization framework, which represents the system objective that we aim to minimize. 
Therefore, problem $\mathcal{P}1$ can be equivalently transformed into

\begin{align}
	\max _{\boldsymbol{s}_m,\boldsymbol{\omega}_m,\boldsymbol{p}_m} &\quad V \times \mathrm{A}_{m} - \Delta(\Theta_m)   \nonumber\\
	\text { s.t. } &\quad\text{(11c), (11d) and (11e)}, \label{4a}
\end{align}

where $V$ is the Lyapunov control parameter.

By decomposing the optimization into slot-wise problems, (\ref{4a}) is transformed into problem $\mathcal{P}1.1$, thus completing the proof.

\section{Poof of Lemma 2 }

Once $s_{n,m}$ and $\omega_{n,m}$ is fixed, problem $\mathcal{P}1.2$ is reduced into the following resource allocation problem:
\begin{align}
	\max _{\boldsymbol{\omega}_m,\tilde{\boldsymbol{p}}_m} &\quad V \sum_{n=1}^N\widehat{\mathrm{A}}_{n,m}({\beta}_{n,m}) - \sum_{n=1}^N Q_{n,m}\tilde{E}_{n,m}  \nonumber\\
	\text { s.t. } &\quad \text{(11c), (11e)}. \label{5}
\end{align}


Latter term $-Q_{n,m}\tilde{E}_{n,m}$ is linear in $\tilde{p}_{n,m}$, and is therefore concave. Term $\widehat{\mathrm{A}}_{n,m}(\beta_{n,m})$ is a concave and non-decreasing function of the transmission ratio $\beta_{n,m}$. For a fixed $\omega_{n,m}$, $\beta_{n,m}$ is:
$$ \beta_{n,m}(\tilde{p}_{n,m}) = C_1 \cdot \log(1 + C_2 \tilde{p}_{n,m}), $$
where $C_1 = \frac{\omega_{n,m} T}{b_\text{total}(s_{n,m})DL^\text{h}_sL^\textbf{w}_s} > 0$ and $C_2 = \frac{h_{n,m}}{\sigma^{2}} > 0$ are constants. The second derivative is 
$$ \frac{\partial^2 \beta_{n,m}}{\partial \tilde{p}_{n,m}^2} = C_1 \cdot \left( \frac{-C_2^2}{\ln(2)(1 + C_2 \tilde{p}_{n,m})^2} \right). $$

Since $C_1, C_2 > 0$, we have $\frac{\partial^2 \beta_{n,m}}{\partial \tilde{p}_{n,m}^2} < 0$ for all $\tilde{p}_{n,m} > 0$.
So $\beta_{n,m}(\tilde{p}_{n,m})$ is a concave function.

The term $\widehat{\mathrm{A}}_{n,m}(\beta_{n,m})$ is the composition of a concave, non-decreasing function and a concave function $\beta_{n,m}$. By composition rules, $\widehat{\mathrm{A}}_{n,m}(\beta_{n,m})$ is concave. Since $V > 0$, $V\widehat{\mathrm{A}}_{n,m}$ is also concave. Therefore, (\ref{5}) is the sum of two concave terms and is therefore concave. The total sum objective is also concave, thus complete the proof.

\section{Proof of Concavity of Problem P2.2}

Let $p = p_{n,m,k}$ be the optimization variable. The objective function $f(p)$ can be written as:
$$ f(p) = K_1 \cdot \log_2(1 + K_2 p) - q_{n,m,k} \cdot p, $$
where $K_1\! =\! v \!\left( \!\frac{\omega_{n,m} T}{D L_s^h L_s^w}\! \right)$ and $K_2 \!= \!\frac{h_{n,m,k}}{\sigma^2}$. Since $v, \omega_{n,m}, T, D, L_s^h, L_s^w, h_{n,m,k}, \sigma^2$ are all positive quantities, $K_1 \!>\! 0$ and $K_2 \!> \!0$.

To check for concavity, we compute the second derivative of $f(p)$ with respect to $p$.
The second derivative $f''(p)$ is:
$$ f''(p) = - \frac{K_1 K_2^2}{\ln(2) (1 + K_2 p)^2}. $$

Since $K_1 > 0$, $K_2 > 0$, $\ln(2) > 0$, and $(1 + K_2 p)^2 > 0$ for all $p > 0$, the numerator is strictly positive and the denominator is strictly positive. Therefore, $f''(p) < 0$ for the entire domain $p > 0$, therefore the objective function $f(p_{n,m,k})$ is strictly concave, thus complete the proof.

\section{Proof of Theorem 1}
\label{sec:appendix_performance_analysis}

We define the Lyapunov function as $L(m)\triangleq\sum_{n=1}^N\frac{1}{2}Q_{n,m}^2$, and the Lyapunov drft as $\Delta_1(m)\triangleq L(m+1)-L(m)$.
The upper bound on the single-round drift-plus-penalty function is given by
\begin{equation} \label{eq:42}
	\setlength\abovedisplayskip{5pt}
	\setlength\belowdisplayskip{5pt}
	\begin{aligned} 
		\Delta_{1}(m) &= L(m+1)-L(m) = \sum_{n=1}^{N}\left(\frac{1}{2}Q_{n,m+1}^{2}-\frac{1}{2}Q_{n,m}^{2}\right) \le \theta_0+\sum_{n=1}^{N}Q_{n,m}\left(E_{n,m}-\bar{E}\right),
	\end{aligned}
\end{equation}
where $\theta_0\triangleq \sum_{n=1}^N \frac{1}{2}\theta_n^2$ and $\theta_n \triangleq \max_m \{|E_{n,m}-\bar{E}_n| \}$. By adding $-V \mathrm{A}_m$ on both sides of (\ref{eq:42}) and define $\xi_m\triangleq \widehat{\mathrm{A}}_m - \mathrm{A}_{m}$, an upperbound on the single-frame drift-plus-penalty function is given by:
\begin{align}
	\Delta_{1}(m)-V\mathrm{A}_{m} &\le \theta_0+\sum_{n=1}^{N}Q_{n,m}\left(E_{n,m}-\bar{E}\right)-V\mathrm{A}_{m} \label{classical}\\
	&=\theta_0+\sum_{n=1}^{N}Q_{n,m}\left(\tilde{E}_{n,m}-\delta_{n,m}-\bar{E}\right)-V\left(\widehat{\mathrm{A}}_{m}-\xi_m\right).\label{estimated}
\end{align}

The classical Lyapunov drift-plus-penalty algorithm is designed to minimize the upper bound of $\Delta_{1}(m) - V \mathrm{A}m$, as expressed in (\ref{classical}). Since the exact values of $E{n,m}$ and $\mathrm{A}_m$ are unavailable, we instead focus on minimizing the estimated drift-plus-penalty, as presented in (\ref{estimated}).

Define the $M$-frame drift as $\Delta_{M} \triangleq L(M+1)-L(1)$. Then the M-frame drift-plus-penalty function can be bounded by:

\begin{align}
	\Delta_{M} -V\sum_{m=1}^{M}\mathrm{A}_{m} &\le \sum_{m=1}^{M}\left(\theta_0+\sum_{n=1}^{N}Q_{n,m}(\tilde{E}_{n,m}-\delta_{n,m}-\bar{E}_n)\right)-V\sum_{m=1}^{M}\mathrm{A}_{m} \nonumber \\
	&= \theta_0M+\sum_{m=1}^{M}\left(\sum_{n=1}^{N}Q_{n,m}\left(\tilde{E}_{n,m}-\bar{E}_{n}\right)-V\mathrm{A}_{m}-\sum_{n=1}^{N}Q_{n,m}\delta_{n,m}\right) \nonumber\\
	&= \theta_0M+\sum_{m=1}^{M}\left(\sum_{n=1}^{N}Q_{n,m}\left(\tilde{E}_{n,m}-\bar{E}_{n}\right)-V\left(\widehat{\mathrm{A}}_{m}-\xi_m\right)-\sum_{n=1}^{N}Q_{n,m}\delta_{n,m}\right).
\end{align}

We denote the optimal offline solution of $\mathcal{P}1$ with superscript $^*$, the solution obtained by classical drift-plus-penalty method for $\mathcal{P}1.1$ with superscript $^\dagger$, and the result produced by our estimated drift-plus-penalty approach for $\mathcal{P}1.2$ with superscript $^\ddagger$.

The drift-plus-penalty over $M$ frames satisfies the following bound:

\begin{align}
	\Delta_{M}^{\dagger}-V\sum_{m=1}^{M}\mathrm{A}_{m}^{\dagger}  & \le \theta_0 M\sum_{m=1}^{M} \left( \sum_{n=1}^{N} Q_{n,m} \left(\tilde{E}_{n,m}-\bar{E}_{n}\right)^\ddagger - V \mathrm{A}_m^\ddagger - \sum_{n=1}^{N} Q_{n,m} \delta_{n,m}^\ddagger \right) \nonumber\\
	& = \theta_0 M+\sum_{m=1}^{M} \left( \sum_{n=1}^{N} Q_{n,m} \left(\tilde{E}_{n,m}-\bar{E}_{n}\right)^\ddagger - V (\widehat{\mathrm{A}}_m^\ddagger - \xi_{n,m}^{\ddagger}) - \sum_{n=1}^{N} Q_{n,m} \delta_{n,m}^\ddagger \right) \nonumber\\
	& \stackrel{(a)}{\le}  \theta_0 M+\sum_{m=1}^{M} \left( \sum_{n=1}^{N} Q_{n,m} \left(\tilde{E}_{n,m}-\bar{E}_{n}\right)^\dagger - V (\widehat{\mathrm{A}}_m^\dagger -\xi_{n,m}^{\ddagger}) - \sum_{n=1}^{N} Q_{n,m} \delta_{n,m}^\ddagger \right) \nonumber\\
	& =  \theta_0 M+\sum_{m=1}^{M} \left( \sum_{n=1}^{N} Q_{n,m} \left( \left({E}_{n,m}-\bar{E}_{n}\right)^\dagger + \delta_{n,m}^\dagger \right) - V\left({\mathrm{A}}_m^\dagger  +\xi_{n,m}^{\dagger}- \xi_{n,m}^{\ddagger}\right)- \sum_{n=1}^{N} Q_{n,m} \delta_{n,m}^\ddagger \right) \nonumber\\
	& = \theta_0 M+\sum_{m=1}^{M} \left( \sum_{n=1}^{N} Q_{n,m} \left({E}_{n,m}-\bar{E}_{n}\right)^\dagger - V \left(\mathrm{A}_m^\dagger +\xi_{n,m}^{\dagger}-\xi_{n,m}^{\ddagger}\right) +\sum_{n=1}^{N} Q_{n,m} \left( \delta_{n,m}^\dagger - \delta_{n,m}^\ddagger \right) \right) \nonumber\\
	&\stackrel{(b)}{\le} \theta_0 M+\sum_{m=1}^{M}\left(\sum_{n=1}^{N}Q_{n,m}\left(E_{n,m}-\bar{E}_n\right)^{*}-V\mathrm{A}_{m}^{*} +2\delta_{0}\sum_{n=1}^{N}Q_{n,m}+ 2\xi_0V\right),\label{8}
\end{align}
where $\delta_0\triangleq \max_{\{n,m\}}\left\{\left|\tilde{E}_{n,m}-E_{n,m} \right|\right\}$ and $\xi_0\triangleq \max_{\{m\}}\left\{\left|\widehat{\mathrm{A}}_m-\mathrm{A}_{m} \right|\right\}$.

Inequality (a) follows from the fact that the optimal solution to $\mathcal{P}1$ attains the minimal objective value $\sum_{n=1}^{N} Q_{n,m}\tilde{E}_{n,m}^{\ddagger} - V \mathrm{A}_m^{\ddagger}$ for each frame $m$. Inequality (b) holds because the drift-plus-penalty procedure minimizes $\sum_{n=1}^{N} Q_{n,m}E_{n,m} - V \mathrm{A}_m$, therefore, substituting the offline optimal policy on the right-hand side can only result in a larger (or equal) value.

Now we bound the right-hand-side of (\ref{8}). Note that $Q_{n,m+1}-Q_{n,m}\le \theta_n, \forall m,n$, and thus
\begin{gather}
	Q_{n,m}=Q_{n,m}-Q_{n,1}=\sum_{\tau=1}^{m-1}\left(Q_{n,\tau+1}-Q_{n,\tau} \right)\leq(m-1)\theta_n, \label{9}\\
	Q_{n,m}\left({E}_{n,m}-\bar{E}_n\right)^*=\left( Q_{n,m}-Q_{n,1}\right)\left({E}_{n,m}-\bar{E}_n\right)^* \le(m-1)\theta_n^2.\label{10}
\end{gather}

Substituting (\ref{9}) and (\ref{10}) into (\ref{8}) yields:
\begin{align}
	\Delta_{M}^{\ddagger}-V\sum_{m=1}^{M}\mathrm{A}_{m}^{\ddagger}  &\le \theta_0 M - V\sum_{m=1}^{M} \mathrm{A}_{m}^{*} +  \sum_{m=1}^{M}\sum_{n=1}^{N}(m+1)\theta_n^2 +2\delta_{0}\sum_{m=1}^{M}\sum_{n=1}^{N}(m-1)\theta_n+ 2\xi_0V \nonumber\\
	&=\theta_0 M - V\sum_{m=1}^{M} \mathrm{A}_{m}^{*} +  \theta_0M(M-1) +M(M-1)\delta_0\sum_{n=1}^N\theta_n+ 2\xi_0V \nonumber\\
	&=-V\sum_{m=1}^{M} \mathrm{A}_{m}^{*}+\theta_0M^2 +M(M-1)\delta_0\sum_{n=1}^N\theta_n+2\xi_0V.\label{11}
\end{align}

Observing that $\Delta_{M}^{\dagger} \ge 0$, the result in (27) of Theorem 1 follows directly from (\ref{11}) after dividing both sides by $V$.


As $\mathrm{A}_{m} \ge 0$, and for $\forall n, \frac{1}{2}Q^2_{n,M+1}\le \Delta_M$, we get
\begin{align}
	\sum_{m=1}^M \left(E_{n,m}-\bar{E}_n \right)\le \sum_{m=1}^M \left( Q_{n,m+1}-Q_{n,m}\right) = Q_{n,M+1}\le \sqrt{2\Delta_M}\le\sqrt{2\theta_0M^2 +2M(M-1)\delta_0\sum_{n=1}^N\theta_n+4\xi_0V}.
\end{align}

Thus eq. (28) in Theorem 1 is proved.

\vfill

\end{document}